\documentclass[final,3p]{elsarticle}
\usepackage{amsmath,amssymb,graphicx}
\usepackage{amsfonts,amsthm}

\usepackage{amsmath,amssymb,graphicx,amsthm}
\usepackage{color}
\usepackage{soul}

\newcommand{\D}{{\mathrm{d}}}
\newtheorem{theorem}{Theorem}

\newtheorem{proposition}{Proposition}
\newtheorem{definition}{Definition}
\newtheorem{corollary}{Corollary}
\newtheorem{remark}{Remark}

\begin{document}

\begin{frontmatter}

\title{{\em Maxallent}: Maximizers of all Entropies and Uncertainty of Uncertainty}
\author{A.~N.~Gorban}
 \ead{ag153@le.ac.uk}
\address{Department of Mathematics, University of Leicester, Leicester, LE1 7RH, UK}

%\maketitle

\begin{abstract}
The entropy maximum approach (Maxent) was developed as a minimization of the subjective
uncertainty measured by the Boltzmann--Gibbs--Shannon entropy.  Many new entropies have
been invented in the second half of the 20th century. Now there exists a rich choice of
entropies for fitting needs. This diversity of entropies gave rise to a Maxent
``anarchism''. The Maxent approach is now the conditional maximization of an appropriate
entropy for the evaluation of the probability distribution when our information is
partial and incomplete. The rich choice of non-classical entropies causes a new problem:
which entropy is better for a given class of applications?  We understand entropy as a
{\em measure of uncertainty which increases in Markov processes.} In this work, we
describe the most general ordering of the distribution space, with respect to which all
continuous-time Markov processes are monotonic (the Markov order). For inference, this
approach results in a {\em set} of conditionally ``most random'' distributions. Each
distribution from this set is a maximizer of its own entropy. This ``uncertainty of
uncertainty'' is unavoidable in the analysis of non-equilibrium systems. Surprisingly,
the constructive description of this set of maximizers is possible. Two decomposition
theorems for Markov processes provide a tool for this description.
\end{abstract}
\date{}

\end{frontmatter}

\section{Introduction\label{sec1}}

Entropy was born in the 19th century as a daughter of energy: $\D S= \delta Q/T$.
Clausius \cite{Clausius1865}, Boltzmann \cite{Boltzmann1872} and Gibbs \cite{Gibbs1875}
(and others) had developed the physical notion of entropy. At the same time, the famous
Boltzmann's formula $S=k \log W$  had opened the informational interpretation of entropy.
In the 20th century, Hartley \cite{Hartley1928} and Shannon \cite{Shannon1948} introduced
a logarithmic measure of information in electronic communication in order ``to eliminate
the psychological factors involved and to establish a measure of information in terms of
purely physical quantities'' (\cite{Hartley1928}, p.~536). Information theory is focused
on entropy as a measure of uncertainty of subjective choice. This understanding of
entropy was returned from information theory to statistical mechanics by Jaynes
\cite{Jaynes1957} as a basis of ``subjective'' statistical mechanics: ``Information
theory provides a constructive criterion for setting up probability distributions on the
basis of partial knowledge, and leads to a type of statistical inference which is called
the maximum entropy estimate. It is least biased estimate possible on the given
information; i.e., it is maximally noncommittal with regard to missing information. That
is to say, when characterizing some unknown events with a statistical model, we should
always choose the one that has Maximum Entropy.'' This is the brief manifesto of the
Maxent (maximum of entropy) methodology.

Entropy is used for measurement of uncertainty in a probability distribution. The Maxent
method finds the maximally uncertain distribution under given values of some moments.
After Jaynes, this approach became very popular in physics \cite{JouEIT2001,GorKar2006},
statistics \cite{Csiszar1991,Gull1988}, econometrics
\cite{GolanJudgeMil1996,JudgeMittelhammer2011} and other disciplines.

The non-classical entropies were invented by R\'enyi \cite{Renyi1961} in the middle of
the 20th century, simultaneously with the expansion of the Maxent approach. This
invention introduced additional uncertainty in the uncertainty evaluation. Maximization
of different entropies produces different probability distributions under the same
conditions. Now, one has to select the proper entropy functional to use in the Maxent
approach. This choice may be non-obvious. The beautiful and transparent understanding of
the Maxent distribution as a unique  ``least biased estimate possible on the given
information'' is now destroyed by the non-classical entropies. If we consider the
non-classical entropies seriously then we have to select the proper entropy for each
problem.

If we do not find solid reasons for the entropy selection then we have to accept this
``Uncertainty of Uncertainty'' (UoU) as the nature of things. In this case, the {\em set
of all the Maxent distributions for different entropies} will evaluate the unknown
``maximally uncertain'' distribution under given conditions. We call this method of
handling the UoU the ``maximization of all entropies'' or {\em Maxallent}. If there are
some reasons for selection of a class of entropy function then we have to select the
conditional maximizer of the entropies from this class.

The widest class of entropies we use in this paper are the {\em Csisz\'ar--Morimoto
conditional entropies ($f$-divergencies)}. They were introduced by R\'enyi in his famous
work \cite{Renyi1961} where he proposed also the ``R\'enyi entropy''. The
$f$-divergencies were studied further by Csiszar \cite{Csiszar1963} and T. Morimoto
\cite{Morimoto1963}. For a discrete probability distribution
$P=(p_i)$ and the positive ``equilibrium distribution'' $P^*=(p^*_i)$, $p^*_i>0$ the
general form of the $f$-divergence is

\begin{equation}\label{Morimoto}
\boxed{H_h(P \| P^*)=\sum_i p^*_i h\left(\frac{p_i}{p_i^*}\right),}
\end{equation}
where $h(x)$ is a convex function defined on the open ($x>0$) or closed ($x\geq 0$)
semi-axis. We use here the notation $H_h(P \| P^*)$ to stress the dependence of $H_h$ both on
$p_i$ and $p^*_i$.

In some practical problems, it is convenient to use a convex function $h(x)$ with
singularity at $x=0$, for example, $h(x)=-\ln x$ (the Burg relative entropy
\cite{Burg1967}). Therefore, we assume that the function $H_h(P \| P^*)$ is defined for
positive $P$ and $P^*$. Convexity of $h(x)$ implies convexity of $H_h(P \| P^*)$ as a
function of $P$. It achieves its minimal value on the equilibrium probability, $P=P^*$
(under conditions $\sum_i p_i=1$, and $p_i >0$).  If $h(x)$ is strictly convex then
$H_h(P \| P^*)$ is also strictly convex and this minimizer (the equilibrium) is unique.

\subsection{Maxallent, approach \#1: parametrization by monotonic function of one variable}

The standard settings for the Maxent approach are: an event space $\Omega$, a divergency
$H_h(P \| P^*)$ and a set of moments $M_r(P)$ ($r=1,\ldots , k$) are given. Here, $P$ is
a probability distribution, $P^*$ is the ``maximally disordered'' probability
distribution (``equilibrium'') and $H_h(P \| P^*)$ measures the deviation of $P$ from
$P^*$. Of course, for general probability spaces we have to assume that $P$ is absolutely
continuous with respect to $P^*$ and that it is possible to compute the divergence $H_h(P
\| P^*)$. The Maxent problem is: for given values of the moments $M_r(P)$ ($r=1,\ldots ,
k$) find the minimizers of $H_h(P \| P^*)$. That is, on the set of probability
distributions with given values of $M_r(P)$ ($r=1,\ldots , k$) find the distributions
that are the closest to the equilibrium $P^*$ if we measure the deviation by $H_h(P\|
P^*)$. The terminological mess ({\em Max}ent and {\em min}imizers) appears due to
historical reasons. Divergences measure the differences between distributions and we
always look for minimizers of them.

To avoid the irrelevant technicalities we consider discrete distributions. Let
$\Omega=\{A_1,A_2,\ldots, A_n\}$ be a finite event space
with probability distributions $P=(p_i)$. The set of probability distribution is the
standard simplex $\Delta^{n-1}$ in $\mathbb{R}^n$. The set of positive distribution
($p_i>0$) is $\Delta^{n-1}_+$, the relative interior of the standard simplex.

The Maxent problem for $H_h(P \| P^*)$ and given values of moments
$\sum_{j}m_{rj}p_j=M_r$ ($r=1,\ldots , k$) reads: find $P\in \Delta^{n-1}$ such that
$$H_h(P \| P^*) \to \min, \; \mbox{ subject to } \; \sum_{j}m_{rj}p_j=M_r.$$
The total probability condition gives $\sum_{j}m_{0j}p_j=1$ ($m_{0j}=1$, $M_0=1$). Assume
that $k+1<n$ and
$${\rm rank} (m_{rj})=k+1 \;\; (r=0,1,\ldots, k; \, j=1,2, \ldots, n).$$
If ${\rm rank} (m_{rj})<k+1$ then just exclude some moments.

The method of Lagrange multipliers gives for $P\in \Delta^{n-1}_+$
\begin{equation}\label{LagrangeMaxEnt}
h'\left(\frac{p_j}{p_j^*}\right)=\sum_{r=0}^k \lambda_r m_{rj}\;\; (j=1,\ldots, n).
\end{equation}
The derivative $h'$ is a monotonic function. Let $h$ be strictly convex. Then the inverse
function $g(y)$ exists, $g(h'(x))=x$ (for positive $x$). We can apply the function $g$ to
both sides of (\ref{LagrangeMaxEnt}) and write the expression of $P$ and the equations
for the Lagrange multipliers $\lambda_r$ that are just the moment conditions $\sum_j
m_{rj}p_j=M_r$:

\begin{equation}\label{LagrangeMaxEntAnsw}
\boxed{\begin{gathered}
p_j=p_j^*g\left(\sum_{r=0}^k\lambda_r m_{rj}\right) \; \;(j=1,\ldots, n); \\
\sum_j m_{\rho j}p_j^*g\left(\sum_{r=0}^k\lambda_r m_{rj}\right)=M_{\rho} \;\; (\rho=0,\ldots, k) \, .
\end{gathered}}
\end{equation}

Therefore, for the class of the strictly convex functions $h$ all the positive {\em
solutions of the Maxent problem for all $f$-divergencies are parameterized
(\ref{LagrangeMaxEntAnsw}) by the monotonic function $g$}.

The function $g$ should be defined on a real interval $(a,b)=h'((0,\infty))$ (it might be
that $a= -\infty$ or $b=\infty$). The image of $g$ should be the real semi-axis
$(0,\infty)$ because $p/p^*$ may be any positive number. Therefore, $\lim_{y\to a}g(y)=0$
and for finite $a$ the function $g$ is defined on $[a,b)$. For each monotonically
increasing function $g$ on a real interval $(a,b)$ with ${\rm im}\, g=(0,\infty)$,  the
corresponding solution of the Maxent problem is given by the distribution
(\ref{LagrangeMaxEntAnsw}), where $\lambda_i$ are the solutions of the corresponding
equation. This solution of the Maxent problem is the conditional minimizer of
$H_h(P\|P^*)$ with $h(x)=\int h'(x) \D x$, where $h'(x)$ is the inverse function of
$g(y)$, i.e. $h'(x)=y$, where $y$ is the solution to the equation $g(y)=x$. The additive
constant in $\int h'(x) \D x$ does not affect the solution of any Maxent problems and may
be chosen arbitrarily. Thus, we present the parametric description of the minimizers of
all strictly convex divergences $H_h(P\|P^*)$. A monotonic function $g$ with the values
range $(0,\infty)$ serves as a parameter in this description.

For the existence of a positive distribution $P$ which satisfies
(\ref{LagrangeMaxEntAnsw}) the moment conditions $\sum_j m_{rj}p_j=M_r$ ($\rho=0,\ldots,
k$) should be compatible with the positivity of $p_i$. Of course, for arbitrary $g$ this
may be not sufficient for the existence of such a positive distribution. To guarantee the
existence of a positive Maxent distribution it is sufficient to add to the function
$h(x)$ a term $\varepsilon x\ln x$ with arbitrarily small positive $\varepsilon$. This
term creates a logarithmic singularity of $h'(x)$ at zero. It is easy to check that this
singularity guarantees the existence of a positive solution of (\ref{LagrangeMaxEntAnsw})
if the moment conditions are compatible with the positivity of $p_i$. For some applied
purposes an additional term $-\varepsilon \ln x$ may be even more convenient
\cite{GorPack2012} because it guarantees the logarithmic singularity of entropy and
$h'(x)$ has the singularity $\sim -1/x$ at zero.

In this paper, the question about existence of the positive Maxent distribution is not
important. We need only the conditions  (\ref{LagrangeMaxEntAnsw}) which are  necessary
and sufficient for a positive distribution $P=(p_i)$ to provide a minimizer of the given
$f$-divergency under moment conditions.

\subsection{Maxallent, approach \#2: the Markov order}

Any Markov process with equilibrium $P^*$ increases disorder. The classical
Boltzmann--Gibbs--Shannon entropy grows in Markov processes. This  theorem (the ``data
processing lemma'') was proved in the first paper of Shannon \cite{Shannon1948} but of
course the entropy growth in kinetics was known before (Boltzmann's $H$-theorem
\cite{Boltzmann1872} and its generalization for the systems without detailed balance
\cite{Boltzmann1887}).

A. R\'enyi proved in the first paper
about the non-classical entropies \cite{Renyi1961} that all $f$-divergencies (\ref{Morimoto}) decrease in
Markov processes with equilibrium $P^*$. Later on, it was demonstrated that this property characterizes
$f$-divergencies among all functions which can be presented in the form of the sums over
states (the ``trace form'') \cite{ENTR3,Amari2009,GorGorJudge2010}.

The generalized data processing lemma was proven \cite{Cohen1993,CohenIwasa1993}: For
every two positive probability  distributions $P,Q$ the divergence $H_h(P \| Q)$
decreases under action of a stochastic matrix $A=(a_{ij})$
$$H_h(AP \| AQ)\leq \overline{\alpha}(A) H_h(P \| Q),$$
where
$$\overline{\alpha}(A)=\frac{1}{2}\max_{i,k} \left\{\sum_j |a_{ij}-a_{kj}| \right\} $$
is the ergodicity contraction coefficient,  $0 \leq \overline{\alpha}(A) \leq 1$.

A second method of handling the UoU is based on a simple remark: ``uncertainty of a
probability distribution should increase in Markov processes''. More precisely, let
the most uncertain distribution $P^*$ be given (the equilibrium).  If a distribution
$P'$ can be obtained from a distribution $P$ in a Markov process with equilibrium $P^*$
then we can assume:
$$\mbox{uncertainty of }P \leq \mbox{uncertainty of }P'.$$
Thus, we do not care about the {\em values} of the uncertainty measure,
we just {\em compare} the uncertainty of distributions:
$P'$ is more uncertain than $P$ under given equilibrium $P^*$
(in this sense, the values vanish but the (pre)order appears \cite{GorGorJudge2010}).

In the Maxent approach, the entropy is used as a (pre)order in the distribution space,
not as a function, and the values are not important because any monotonically increasing
transformation of the entropy does not change the solution of the Maxent problem. Of
course, in some other applications the values of entropy are important: in coding theory
(bits per symbol) and in thermodynamics ($\D U=T\D S$) the values of the entropy have a
specific important sense. Nevertheless, when we discuss the entropy as a measure of
uncertainty and work with the huge population of  non-classical entropies, these
entropies are, in their essence, (pre)orders on the space of distributions.

We consider the {\em continuous time Markov processes with a given equilibrium
distribution $P^*$}. By definition, the equilibrium is the unconditionally maximally
uncertain distribution. To add the moment conditions we define a linear manifold in the
space of distributions.  For every non-equilibrium distribution $P$ each Markov process
with the equilibrium distribution $P^*$ determines the direction of $P$ evolution, $\D P/
\D t$. In this direction, the distribution becomes more uncertain. Let us take this
property as a definition of the uncertainty. Instead of an entropy functional we use the
transitive closure of this relation, define an order on the space of distributions and
call it the ``Markov order'' \cite{GorGorJudge2010}.

Let $\mathbf{Q}(P,P^*)$ be a cone of possible time derivatives $\D P/ \D t$ for a given
probability distribution $P$, the equilibrium $P^*$, and all Markov processes with
equilibrium  $P^*$.

For fixed values of moments, $M_r$, the conditionally linear manifold $L$ in the space of the
probability distributions is given by equations $\sum_{j}m_{rj}p_j=M_r$ ($r=0,\ldots,
k$). We can consider $P^0 \in L$ as a possibly extremely disordered distribution on $L$,
if for any Markov process with equilibrium $P^*$ the solution $P(t)$ of the Kolmogorov equation
with initial condition $P(0)=P^0$ has no points on the conditionally linear manifold
$L$ for $t> 0$ (we assume that $P^0$ is not a steady state for this process). Instead of
this global condition, we consider the local condition (Fig.~\ref{CondExtr}).
\begin{definition}\label{Def:localOrder}The distribution $P^0 \in L \cap \Delta^{n-1}_+$ is a local minimum of the Markov order on $L\cap \Delta^{n-1}_+$ if
\begin{equation}\label{localExtr}
\boxed{(P^0+\mathbf{Q}(P^0,P^*))\cap L= \{P^0\}.}
\end{equation}
\end{definition}
Further, for short, we can omit $\Delta^{n-1}_+$ and call $P^0$ ``a local minimum of the
Markov order on $L$''. In this definition, we substitute the trajectories $P(t)$ by their
tangent directions at point $P^0$, $\D P(t)/\D t \in \mathbf{Q}(P^0,P^*)$. In
Sec.~\ref{Sec:order} we justify this substitution and prove that the local condition
(\ref{localExtr}) holds if and only if for every Markov process with equilibrium $P^*$
the solution  $P(t)$ of the Kolmogorov equation with initial condition $P(0)=P^0$ has no
points on the condition linear manifold $L$ for $t>0$ (if $P^0$ is not a steady state for
the process).

\begin{figure}
\centering{
\includegraphics[width=0.4\textwidth]{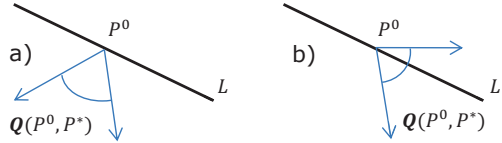}
\caption{\label{CondExtr}The local condition (\ref{localExtr}): $P^0$ may be an extremely disordered distribution on the condition linear manifold $L$ if the set
$P^0+\mathbf{Q}(P^0,P^*)$ intersects the linear manifold of conditions at the only point $P^0$; (a)$(P^0+\mathbf{Q}(P,P^*))\cap L= \{P^0\}$ and $P^0$ may be an extremely disordered distribution on $L$; (b) $(P^0+\mathbf{Q}(P,P^*))\cap L \varsupsetneqq \{P^0\}$ and $P^0$ has no chance to be an extremely disordered distribution on $L$. In case (b), there are more disordered distributions on $L$ achievable by the Markov processes from the initial distribution $P^0$.} }
\end{figure}

For applications, we need the local minima condition formalized by
Definition~\ref{Def:localOrder} and the local order generated by the cone
$\mathbf{Q}(P^0,P^*))$ only. The general notion of (global) Markov order appears later,
in Section~\ref{Sec:Equival}, where we prove equivalence of the Maxima of all entropies
and the Markov order approaches. Surprisingly, the set of the conditional minimizers of
all $f$-divergencies  and the set of the conditionally minimal elements of the Markov
order coincide for the same conditions (Sec.~\ref{Sec:Equival}). These sets include all
reasonable hypotheses about conditionally most uncertain distributions. Let us call the
problem of description of all the conditional minima of the Markov order the {\em
Maxallent problem}.

\subsection{Main tool: decomposition theorems}

The main tools for constructive work with the Markov orders are the decomposition
theorems for Markov chains. {\em The first decomposition theorem} states that every
Markov chain with a positive equilibrium distribution is a convex combination of the
simple directed cyclic Markov chains with the same equilibrium. The coefficients in this
decomposition do not depend on the current probability distribution: the vector field $\D
P/ \D t$ for a general Markov chain is a convex combination of these vector fields for
simple cyclic Markov chains with the same positive equilibrium.

{\em The second decomposition theorem} states that for every Markov chain
with a positive equilibrium distribution and for any non-equilibrium distribution $P$ the
velocity vector $\D P/ \D t$ is a convex combination of the velocity vectors for the
simple cyclic Markov chains {\em of the length two} with the same equilibrium (i.e. of the
reversible transitions between two states, $A_i \rightleftharpoons A_j$). The coefficients in this decomposition
typically depend on the current probability distribution.

The idea of the first decomposition theorem was used by Boltzmann in 1882
\cite{Boltzmann1887} in his proof of the $H$-theorem for systems without detailed
balance. (This was his answer to the Lorentz objections \cite{Lorentz1887}.) He did not
formulate this theorem separately but efficiently used the cycle decomposition for
generalization of detailed balance. Later on, his extension of the detailed balance
conditions were analyzed by many authors under different names as ``cyclic balance'',
``semi-detailed balance'' or ``complex balance'' (see, for example, the review
\cite{GorbanShahzad2011}). Now, the theory of the cycle decomposition is a well developed
area of the theory and applications of the random processes \cite{Kalpazidou2006}.

The second decomposition theorem is less known. We found this theorem in the analysis of
the Markov order \cite{GorGorJudge2010}. This decomposition means that for the general
first-order kinetics and an arbitrary non-equilibrium probability distribution $P$  there
exists a system with detailed balance and the same equilibrium that has the same velocity
$dP/dt$ at point $P$ \cite{GorbanEq2012arX}: the classes of the general Markov processes
and the Markov processes with detailed balance are pointwise equivalent.

The decomposition theorems are discussed in Appendix~B in more detail.

\section{Local minima of Markov order \label{Sec:order}}

Let us consider continuous time Markov chains with $n$ states $A_1, \ldots , A_n$. The
{\em Kolmogorov equation (or master equation)} for the probability distribution $P=(p_i)$
is
\begin{equation}\label{MAsterEq0}
\frac{\D p_i}{\D t}= \sum_{j, \, j\neq i} (q_{ij}p_j-q_{ji}p_i) \;\; (i=1,\ldots, n) ,
\end{equation}
where $q_{ij}$ ($i,j=1,\ldots, n$, $i\neq j$) are non-negative.

In this notation, $q_{ij}$ is the {\em rate constant} for the transition $A_j \to A_i$.
Any  non-negative values of the coefficients $q_{ij}$ ($i\neq j$) correspond to a master
equation. Therefore, the set of all the Kolmogorov equations (\ref{MAsterEq0}) may be
considered as the positive orthant $\mathbb{R}_+^{n(n-1)}$ in $\mathbb{R}^{n(n-1)}$ with
coordinates $q_{ij}$ ($i\neq j$).

Now, let us restrict our consideration to the set of the Markov chains with the
given positive equilibrium distribution $P^*$ ($p^*_i>0$).
\begin{equation}\label{MasterEquilibrium}
\sum_{j, \, j\neq i} q_{ij}p^*_j = \left(\sum_{j, \, j\neq i}
q_{ji}\right)p^*_i \; \mbox{ for all }i=1,\ldots, n.
\end{equation}
This system of uniform linear equations define a cone of the  $q_{ij}$
($i,j=1,\ldots, n$, $i\neq j$) in $\mathbb{R}_+^{n(n-1)}$.

Under the {\em balance condition} (\ref{MasterEquilibrium}), the Kolmogorov equations
(\ref{MAsterEq0}) may be rewritten in a convenient  equivalent form:
\begin{equation}\label{MAsterEq1}\boxed{
\frac{\D p_i}{\D t}= \sum_{j, \, j\neq i}
q_{ij}p^*_j\left(\frac{p_j}{p_j^*}-\frac{p_i}{p_i^*}\right) \;\; (i=1,\ldots, n) .}
\end{equation}

We use below one of the $f$-divergencies (\ref{Morimoto}) with $h(x)=(x-1)^2$. It is a
quadratic divergence, the weighted $l_2$ distance between $P$ and $P^*$:
$$H_2(P\|P^*)=\sum_i\frac{(p_i-p^*_i)^2}{p^*_i}\, .$$

With the master equation in the form (\ref{MAsterEq1}), it is straightforward to
calculate the time derivative of $H_2(P\|P^*)$
\begin{equation}\label{QuadLyap}
\frac{\D H_2(P\|P^*)}{\D t}=-\sum_{i,j, \, j\neq i}q_{ij}p^*_j\left(\frac{p_i}{p^*_i} - \frac{p_j}{p^*_j}\right)^2 \leq 0\, .
\end{equation}
Each term in the sum is non-negative. The time derivative (\ref{QuadLyap})
is {\em strictly negative} if for a transition $A_j \to
A_i$ the rate constant is positive, $q_{ij}>0$, and $\frac{p_i}{p^*_i} \neq
\frac{p_j}{p^*_j}$. Hence, if the state $P$ is not an equilibrium (i.e., the
right hand side in (\ref{MAsterEq1}) is not zero) then $\frac{\D
H_2(P\|P^*)}{\D t}<0$.

An important class of the Markov chains is formed by reversible chains with detailed
balance. The {\em detailed balance} condition reads:
\begin{equation}\label{detBal}
q_{ij}p^*_j=q_{ji}p^*_i \;\; \mbox{ for all }i,j=1,\ldots, n.
\end{equation}
Under this condition, there are only $\frac{n(n-1)}{2}$ independent coefficients among
$n(n-1)$ numbers $q_{ij}$. For example, we can arbitrarily select $q_{ij} \geq 0$ for
$i>j$ and then take $q_{ij} = q_{ji}\frac{p_i^*}{p^*_j}$ for $i<j$. So,  for given $P^*$,
the cone of the detailed balance systems (\ref{detBal}) is a positive orthant in
$\mathbb{R}^{\frac{n(n-1)}{2}}$ embedded in $\mathbb{R}_+^{n(n-1)}$. The equilibrium
fluxes $$w_{ij}^*=q_{ij}p^*_j=q_{ji}p^*_i \;\; (i>j)$$ are the convenient coordinates in
$\mathbb{R}^{\frac{n(n-1)}{2}}$ for a description of the systems with detailed balance.

Let $\mathbf{Q} (P,P^*)$ be the set of all possible velocities $\D P/\D t$ at a
non-equilibrium distribution $P$ for all Markov chains which obey a given positive
equilibrium $P^*$. According to the second decomposition theorem, the set of all possible
velocities $\D P/\D t$ for the chains with detailed balance and the same equilibrium is
the same cone $\mathbf{Q} (P,P^*)$. Therefore, $\mathbf{Q} (P,P^*)$ is a convex
polyhedral cone and its extreme rays consist of the velocity vectors for  two-state
Markov chains $A_i \rightleftharpoons A_j$ with rate constants $q_{ji}= \kappa /p_j^*$,
$q_{ji}= \kappa /p_i^*$ ($\kappa >0$).

The construction of the cones of possible velocities was proposed in 1979
\cite{Gorban1979} for systems with detailed balance in the general setting, for nonlinear
chemical kinetics. These systems are represented by stoichiometric equations of the
elementary reactions coupled with the reverse reactions:
\begin{equation}\label{Stoichiometric}
\alpha_{\rho 1}A_1+\ldots + \alpha_{\rho n}A_n \rightleftharpoons \beta_{\rho 1}A_1+\ldots
+ \beta_{\rho n}A_n\, ,
\end{equation}
where $\alpha_{\rho i}, \, \beta_{\rho i} \geq 0$ are the stoichiometric coefficient,
$\rho $ is the reaction number ($\rho =1, \ldots, m$). The {\em stoichiometric vector} of
the $\rho $th reaction is an $n$ dimensional vector $\gamma_\rho $ with coordinates
$\gamma_{\rho i} = \beta_{\rho i}- \alpha_{\rho i}$. The reaction rate is
$w_{\rho}=w_{\rho}^{+}-w_{\rho}^-$, where $w_{\rho}^{+}$ is the rate of the direct
elementary reaction and $w_{\rho}^-$ is the rate of the reverse reaction

The equilibria of the $\rho $th pair of reactions (\ref{Stoichiometric}) form a
hypersurface in the space of concentrations. The intersection of these surfaces for all
$\rho $ is the equilibrium (with detailed balance). Each surface  of the equilibria of a
pair of elementary reactions (\ref{Stoichiometric}) divides the non-negative orthant of
concentrations into three sets: (i) $w_{\rho}>0$, (ii) $w_{\rho}=0$ (the surface of the
equilibria) and (iii) $w_{\rho}<0$. All the surfaces of equilibria ($w_{\rho}=0$) divide
the non-negative orthant of concentrations into compartments. In each compartment, the
dominant direction of each reaction (\ref{Stoichiometric}) is fixed and, hence, the cone
of possible velocities is also constant. It is a piecewise constant function of
concentrations:
$$\mathbf{Q}={\rm cone}\{\gamma_{\rho } {\rm sign}(w_{\rho }) \, | \, \rho =1, \ldots ,
m\}\, ,$$ where ``cone" stands for the conic hull, that is the set of all linear
combinations with non-negative coefficients. Here and below we use the three-valued sign
function (with values $\pm 1$ and 0).

Let us apply this construction to Markov chains with detailed balance. Let us join
the transitions $A_i \rightleftharpoons A_j$ in pairs (say, $i>j$) and introduce the {\em
stoichiometric vectors} $\gamma^{ji}$ with coordinates:
\begin{equation}
\gamma^{ji}_k=\left\{\begin{array}{ll}
-1 &\mbox{ if } k=j,\\
1 &\mbox{ if } k=i,\\
0 &\mbox{ otherwise}.
\end{array}
\right.
\end{equation}
Let us rewrite the Kolmogorov equation for the Markov process with detailed balance
(\ref{detBal}) in the quasichemical form:
\begin{equation}\label{QuasiChemKol}
\frac{\D P}{\D t}=\sum_{i>j}w_{ij}^*\left(\frac{p_j}{p_j^*}-\frac{p_i}{p_i^*}\right) \gamma^{ji}\, .
\end{equation}
Here, $w_{ij}^*=q_{ij}p^*_j=q_{ji}p^*_i$ is the {\em equilibrium flux} from $A_i$ to
$A_j$ and back.

The cone of possible velocities for (\ref{QuasiChemKol}) is
\begin{equation}\label{COneQPP}
\boxed{\mathbf{Q}(P,P^*)={\rm cone}\left\{\gamma^{ji}
{\rm sign}\left(\frac{p_j}{p_j^*}-\frac{p_i}{p_i^*}\right) \, \Big| \, i>j \right\}\, .}
\end{equation}

The standard simplex of distributions $P$ is divided by linear manifolds
$\frac{p_i}{p_i^*}=\frac{p_j}{p_j^*}$ into compartments. They are the polyhedra where the
cone of the local Markov order  $\mathbf{Q}{(P,P^*)}$ is constant. The compartments for
the Markov chains with the positive equilibrium $P^*$ correspond to various partial
orders on the finite set $\{p_i/p_i^*\}$ ($i=1, \ldots, n$).

Let us describe the compartments and cones in more detail following
\cite{GorGorJudge2010}. For every natural number $k \leq n-1$ the $k$-dimensional
compartments are enumerated by surjective functions $\sigma :\{1,2,\ldots ,n\} \to
\{1,2,\ldots ,k+1\}$. Such a function defines the partial ordering of quantities
$\frac{p_j}{p_j^*}$ inside the compartment:
\begin{equation}
\frac{p_i}{p_i^*} > \frac{p_j}{p_j^*} \; \; {\rm if} \;\; \sigma(i)
< \sigma (j); \;\;\; \frac{p_i}{p_i^*} = \frac{p_j}{p_j^*} \; \;
{\rm if} \;\; \sigma(i) = \sigma (j)\, .
\end{equation}
Let  $\mathcal{C}_{\sigma}$ be the corresponding compartment and  $Q_{\sigma}$ be the
corresponding local Markov order cone ($\mathbf{Q}(P,P^*)=Q_{\sigma}$ if $P\in
\mathcal{C}_{\sigma}$).

For a given surjection $\sigma$ the compartment $\mathcal{C}_{\sigma}$ and the cone
$Q_{\sigma}$ have the following description:
\begin{equation}\label{CompartConeDescr}
\boxed{\begin{gathered}
\mathcal{C}_{\sigma}=\left\{ P \ \Big| \ \frac{p_i}{p^*_i}=\frac{p_j}{p^*_j} \;\; {\rm
for} \;\; \sigma(i)=\sigma(j) \;\;{\rm and} \;\; \frac{p_i}{p^*_i}>\frac{p_j}{p^*_j}
\;\; {\rm for} \;\; \sigma(j)=\sigma(i)+1 \right\}\, ;\\
Q_{\sigma}={\rm cone}\{\gamma^{ij} \ | \ \sigma(j)=\sigma(i)+1 \}.
\end{gathered}}
\end{equation}

In Fig.~\ref{3stateCones}, the partition of the standard distribution simplex into
compartments, and the cones (angles) of possible velocities are presented for the Markov
chains with three states. In the construction of this cone, reversible chains with
detailed balance are used. Due to the second decomposition theorem, this construction of
the cone of possible velocities is valid for the class of general Markov chains (and not
only for reversible chains) with the same equilibrium. It seems quite surprising that the
Markov order for general Markov chains is generated by the reversible Markov chains which
satisfy the detailed balance principle.

\begin{figure}[t]
\centering{
\includegraphics[width=0.6\textwidth]{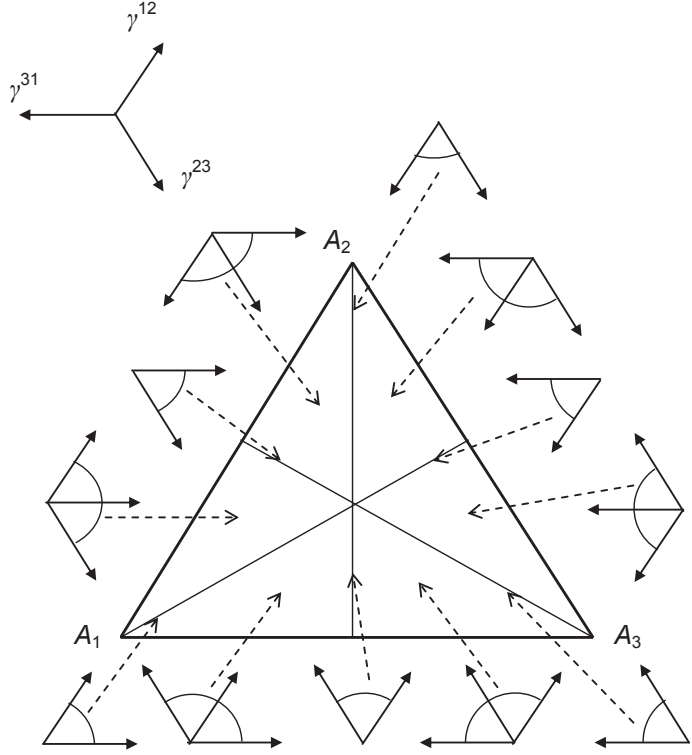}
\caption{\label{3stateCones}The cones of possible velocities  $\mathbf{Q}$ for all Markov chains
with three states and equilibrium equidistribution ($p_i^*=1/3$). The triangle of the probability distributions $(p_1,p_2,p_3)$
($p_i\geq0$, $p_1+p_2+p_3=1$) has the vertices $A_i$, where $p_i=1$ and other probabilities are zeros.
Equilibrium is the centre of the triangle.
This triangle is divided by three lines of partial equilibria ($A_i\rightleftharpoons A_i$) into 12 compartments and the equilibrium point.
Six compartments are triangles and  six other compartments are segments. For all compartments the cones (here the angles) of possible velocities are shown.
Each cone is connected with the corresponding compartment by a dashed line. In each cone, the vectors
$\gamma^{ji} {\rm sign}\left(\frac{p_j}{p_j^*}-\frac{p_i}{p_i^*}\right)$ are presented. For the 2D (triangle) compartments
all three vectors are non-zero. For the 1D compartments (segments) only two these vectors are non-zero. The vectors $\gamma^{ji}$
are presented separately, in the top left corner.} }
\end{figure}

Let $L$ be a linear manifold in  the probability distribution space. Due to
Definition~\ref{Def:localOrder}, $P^0\in L \cap \Delta^{n-1}_+$ is a local minimum of the
Markov order on $L \cap \Delta^{n-1}_+$ if the condition (\ref{localExtr}) holds.

In Fig.~\ref{Fig:extremes1} the sets of conditional minimizers are presented for the
Markov order on the straight line $L$ for the Markov chain with three states and
symmetric equilibrium ($p_i^*=1/3$).  Two general positions of $L$ in the probability
triangle are used (Fig.~\ref{Fig:extremes1}a,b). If $L$ is parallel to one side of the
triangle (Fig.~\ref{Fig:extremes1}c) then the moments are just some of the $p_i$ and the
extreme points of the Markov order on $L\cap \Delta^{n-1}_+$ coincide with the partial
equilibria.

\begin{figure}
\centering{
\includegraphics[width=0.9\textwidth]{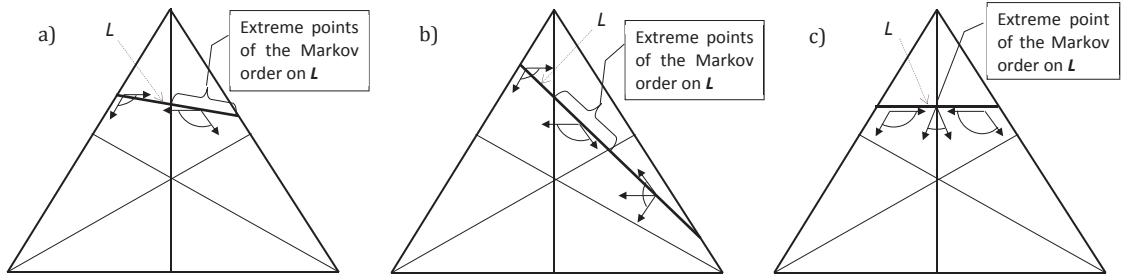}
\caption{\label{Fig:extremes1}a) and b) Extreme points of the Markov order for the Markov chain with three states
and different positions of the condition line $L$. c) Extreme points of the Markov order
coincide with the partial equilibria, when the moments are just some of $p_i$} }
\end{figure}

Let $J$ be a set of pairs of indexes $(i,j)$ ($i > j$) and $\mathcal{K}_J$ be the class
of kinetic equations (\ref{QuasiChemKol}) with $w_{ij}^*=0$ for $(i,j) \notin J$  and
$w_{ij}^*\geq 0$ for $(i,j) \in J$ ($i\neq j$). We  define $\Phi_J(P^0)$ for an initial
distribution $P^0$ as a set of all values $P(t)$ ($t>0$) for solutions $P(t)$ of all
equations from the class $\mathcal{K}_J$ with initial value $P(0)=P^0$.

Consider a cone of possible velocities for the set of transitions $A_i \rightleftharpoons
A_j$, $(i,j)\in J$:
$$\mathbf{Q}_J(P,P^*)={\rm cone}\left\{\gamma^{ji}
{\rm sign}\left(\frac{p_j}{p_j^*}-\frac{p_i}{p_i^*}\right) \, \Big| \,(i,j)\in J \right\}\, .$$

The following proposition states that in a vicinity of the distribution $P^0$ the sets
$\Phi_J(P^0)$ and $P^0 + \mathbf{Q}_J(P^0,P^*)$ coincide. This gives a justification of
the use of the cone of the tangent directions $\mathbf{Q}_J(P^0,P^*)$  in the definition
of the local minima of the Markov order (\ref{localExtr}).

\begin{proposition}\label{OrderJustify1}
Let $\frac{p_j}{p_j^*}-\frac{p_i}{p_i^*}\neq 0$ ($(i,j) \in J$) for a distribution $P=P^0$. There exists a vicinity $U$ of $P^0$ where $P^0+\mathbf{Q}_J(P^0,P^*)$ coincides with $\Phi_J(P^0)$:
$$(P^0+\mathbf{Q}_J(P^0,P^*))\cap U=\Phi_J(P^0)\cap U\, .$$
\end{proposition}
\begin{proof}
There exists a Euclidean ball $B_r$ around $P^0$ where
$\frac{p_j}{p_j^*}-\frac{p_i}{p_i^*}\neq 0$ ($(i,j) \in J$). Due to (\ref{QuadLyap}),
inside $B_r$, the divergence $H_2(P\|P^*)$ strictly decreases with $\lambda$ increasing
along any ray $P^0+\lambda e$, $e\in \mathbf{Q}_J(P,P^*)$ ($\lambda>0$). For each ray, we
can find the minimum of $H_2(P\|P^*)$ in $B_r$. Let the maximum of these minima be $h_r
(P^0)$:
$$h_r (P^0) =\max_{e\in \mathbf{Q}_J(P,P^*)}\left\{ \min_{\lambda >0} \{H_2(P^0+\lambda e\|P^*) \, | \,  P^0+\lambda e \in B_r \} \right\}\, .$$
By construction, $ H_2(P^0\|P^*)> h_r (P^0)$. The set $$U=\{P \in B_r \, | \,
H_2(P\|P^*)> h_r (P^0)\}$$ is a vicinity of $P^0$. The intersection
$(P^0+\mathbf{Q}_J(P^0,P^*))\cap U$ is
$$(P^0+\mathbf{Q}_J(P^0,P^*))\cap U=\{P\in (P^0+\mathbf{Q}_J(P^0,P^*)) \, | \, H_2(P\|P^*)> h_r (P^0)\}\, .$$
For any system from $\mathcal{K}_J$ on $B_r$ and for any distribution $P\in
(P^0+\mathbf{Q}_J(P^0,P^*))$ the velocity vector $\D P/\D t$ belongs to
$\mathbf{Q}_J(P^0,P^*)$. Obviously, $(P+\mathbf{Q}_J(P^0,P^*))\subset
(P^0+\mathbf{Q}_J(P^0,P^*))$.  Therefore, the solution of this system with the initial
condition $P(0)=P^0 $ may leave the intersection $(P^0+\mathbf{Q}_J(P^0,P^*))\cap U$
through the level surface $ H_2(P^0\|P^*)= h_r (P^0)$ only. After that, the solution
cannot return to $U$ because in $U$ the value of $H_2(P^0\|P^*)$ are bigger, $
H_2(P^0\|P^*)> h_r (P^0)$, and $H_2(P(t)\|P^*)$ should decrease in time along every
solution of any system from $\mathcal{K}_J$. Thus, one inclusion is proven,
$$(P^0+\mathbf{Q}_J(P^0,P^*))\cap U\supseteq \Phi_J(P^0)\cap U\, .$$

To prove the second inclusion, $(P^0+\mathbf{Q}_J(P^0,P^*))\cap U\subseteq
\Phi_J(P^0)\cap U$, we have to demonstrate that the solutions $P(t)$ ($P(0)=P^0$,
$t\geq 0$) of the equations from $\mathcal{K}_J$ cover $(P^0+\mathbf{Q}_J(P^0,P^*))$ in
some vicinity of $P^0$.

The polyhedral cone $\mathbf{Q}_J(P^0,P^*)$ is covered by the simplicial cones spanned by
the sets of linearly independent vectors $\gamma^{ji} {\rm
sign}(\frac{p_j^0}{p_j^*}-\frac{p_i^0}{p_i^*})$. Therefore, it is sufficient to prove the
second inclusion for the simplicial cones $\mathbf{Q}_J(P^0,P^*)$.

Let the vectors $\{\gamma^{ji} | (i,j)\in J\}$ be linearly independent. For the
simplicity of notation, let us enumerate the states in the order of the values of
${p_j^0}/{p_j^*}$:
$$\frac{p_1^0}{p_1^*}\geq \frac{p_2^0}{p_2^*} \geq \ldots \geq \frac{p_n^0}{p_n^*}\, .$$
In these notations, ${\rm sign}(\frac{p_j^0}{p_j^*}-\frac{p_i^0}{p_i^*})=1$ for all
$(i,j)\in J$ because $i>j$ and $\frac{p_j^0}{p_j^*}-\frac{p_i^0}{p_i^*}\neq 0$ for
$(i,j)\in J$.

Consider a subset of the cone $\mathbf{Q}_J(P^0,P^*)$ (a ``pyramid'')

\begin{equation}\label{pyramid}
\mathcal{Q}_J(P^0,P^*)=\left\{\sum_{(i,j) \in J} \theta_{ij}\gamma^{ji}\, \Big| \,
\theta_{ij}\geq 0, \, \sum_{(i,j)\in J} \theta_{ij} < 1  \right\}\, .
\end{equation}

The ``base'' of this pyramid is a simplex

$$\mathcal{B}_J(P^0,P^*)=\left\{\sum_{(i,j)\in J} \theta_{ij}\gamma^{ji}\,
\Big| \, \theta_{ij}\geq 0, \, \sum_{(i,j)\in J} \theta_{ij} = 1 \right\}\, .$$

Let $\alpha>0$ be sufficiently small and, therefore,
$\frac{p_j}{p_j^*}-\frac{p_i}{p_i^*}\neq 0$ ($(i,j) \in J$) in $P^0 + \alpha
\mathcal{Q}_J(P^0,P^*)$. For this $\alpha$, a  solution $P(t)$ ($t\geq 0$) of an equation
from the class $\mathcal{K}_J$ with initial data $P(0)=P^0$ may leave $P^0 + \alpha
\mathcal{Q}_J(P^0,P^*)$ only through its base, $P^0+ \alpha \mathcal{B}_J(P,P^*)$.

Let us prove that if $\alpha$ is sufficiently small then for each point $y \in
\mathcal{B}_J(P,P^*)$ there exists a system in $\mathcal{K}_J$ whose solution $P(t)$ ($t>
0$, $P(0)=P^0$) leaves $P^0 + a\mathcal{Q}_J(P^0,P^*)$ through the point $P^0+\alpha y$.
This means that $P(t_1)= P^0+ \alpha x $ for some $t_1>0$ and $P(t) \in
\mathcal{Q}_J(P,P^*)$ for $0<t< t_1$.

Each vector $x\in \mathcal{B}_J(P,P^*)$ can be expanded into a linear combination of
$\gamma^{ji}$ ($(i,j)\in J$):
\begin{equation}\label{x-representation}
x=\sum_{(i,j)\in J}\theta_{ij} \gamma^{ji}, \; \theta_{ij}\geq 0 \mbox{ and } \sum_{(i,j)\in
J} \theta_{ij} = 1 .
\end{equation}
With this expansion we define the system $K_x\in \mathcal{K}_J$ by the condition
$\frac{\D P}{\D t}\big|_{P=P^0}=x$:
\begin{equation}\label{x-equation}
\frac{\D P}{\D t}=\sum_{(i,j)\in J}\theta_{ij}\left(\frac{p_j^0}{p_j^*}-
\frac{p_i^0}{p_i^*}\right)^{-1} \left(\frac{p_j}{p_j^*}-\frac{p_i}{p_i^*}\right) \gamma^{ji}\, .
\end{equation}
(Just take $w_{ij}^*=\theta_{ij}\left(\frac{p_j^0}{p_j^*}-\frac{p_i^0}{p_i^*}\right)$ for
$(i,j)\in J$ in (\ref{QuasiChemKol}).) A solution $P(t)$ ($P(0)=P^0$) of this equation
(\ref{x-equation}) can be also expanded into a linear combination of $\gamma^{ji}$
($(i,j)\in J$) ($x \in \mathcal{B}_J(P,P^*)$):
\begin{equation}\label{x-solution}
P(t)=P^0+tx+\frac{t^2}{2}f(t,x)=P^0+ \sum_{\theta_{ij}> 0}\left(t\theta_{ij}+ \frac{t^2}{2} \nu_{ij}(t,\{\theta_{lm}\}) \right) \gamma^{ji},
\end{equation}
where $\nu_{ij}(t,\{\theta_{lm}\})$ are analytic functions. If $x$ belongs to a face $F$
of the cone $\mathbf{Q}_J(P^0,P^*)$ then $P(t) \in P^0+F$ for sufficiently small $t$.

The moment $t=t(\alpha,x)$ when the solution $P(t)$ (\ref{x-solution}) leaves $P^0 +
\alpha \mathcal{Q}_J(P^0,P^*)$ is a root of equation
$$t+\frac{t^2}{2}\sum_{(i,j)\in J} \nu_{ij}(t,\{\theta_{lm}\}) = \alpha .$$
Due to the standard inverse function theorems this root exists and the function
$t(\alpha,x)$ is  smooth for sufficiently small $\alpha$ for all $x\in
\mathcal{B}_J(P,P^*)$, and $t(\alpha,x)=\alpha+o(\alpha)$. The solution $P(t)$
(\ref{x-solution}) of the system $K_x$ (\ref{x-equation}) leaves $P^0 + \alpha
\mathcal{Q}_J(P^0,P^*)$ at the point $P(t(\alpha,x))=P^0+\alpha y(x)$, where $y(x)\in
\mathcal{B}_J(P,P^*)$.

To prove that $y(\bullet): \mathcal{B}_J(P,P^*) \to \mathcal{B}_J(P,P^*)$ is a
homeomorphism of the simplex $\mathcal{B}_J(P,P^*)$ onto itself, let us notice that the
map $x \mapsto y(x)$ leaves the faces of the simplex $\mathcal{B}_J(P,P^*)$ invariant:
vertices transform into themselves, the same for edges, etc.

We use the following topological lemma, the {\em multidimensional intermediate value
theorem}. Consider a continuous map  $\Psi: \Delta_n \to \Delta_n$ of the $n$-dimensional
standard simplex into itself. Let each face $F\subset\Delta_n$ be $\Psi$-invariant, i.e.
$\Psi (F) \subset F$. Then $\Psi$ is surjective. The proof is possible by induction in
$n$: for $n=0$ it is obvious, for $n=1$ this is just a 1D intermediate value theorem. In
all dimensions, it can be proved on the basis of the ``no-retraction theorem''
\cite{Kuratowski} and simple inductive topological reasoning, which reduces the general
case to the situation when all the faces $F\subset\Delta_n$ consist of fixed points of
the map $\Psi$.

Therefore, for sufficiently small $\alpha$ the solutions $P(t)$ ($P(0)=P^0$, $t\geq 0$)
of the equations from $\mathcal{K}_J$ cover $(P^0+a\mathcal{Q}_J(P^0,P^*))$ in some
vicinity of $P^0$. The second inclusion is proven. Let us combine the inclusions and
reduce the vicinities, if necessary.
\end{proof}

If $\frac{p^0_i}{p^*_i}\neq  \frac{p^0_j}{p^*_j}$ for all $i,j$ ($i\neq j$) then
$$\mathbf{Q}(P^0,P^*)=\mathbf{Q}(P,P^*)$$
for $P$ in some vicinity of $P^0$. If for some pairs $i,j$ ($i\neq j$)
$\frac{p^0_i}{p^*_i} =  \frac{p^0_j}{p^*_j}$ (see Fig.~\ref{Fig:extremes1}c) then for
some $P\in P^0+ \mathbf{Q}(P^0,P^*)$ the cone $\mathbf{Q}(P,P^*)$ may be bigger than
$\mathbf{Q}(P^0,P^*)$ even in a small vicinity of $P^0$. Nevertheless, the set of
trajectories $P(t)$ ($t>0$, $P(0)=P^0$) remains in $P^0+ \mathbf{Q}(P^0,P^*)$  for
sufficiently small $t$. Let us prove this statement.

Let $\mathcal{K}$ be the class of all master  equations with detailed balance with the
positive equilibrium $P^*$  (\ref{QuasiChemKol}) with $w_{ij}^*\geq 0$ for all $(i,j)$
($i>j$). We define $\Phi(P^0)$ for an initial distribution $P^0$ as a set of all values
$P(t)$ ($t>0$) for solutions $P(t)$ of all equations from the class $\mathcal{K}$ with
initial value $P(0)=P^0$.
\begin{proposition}\label{OrderJustify2}
For every probability distribution $P^0$ there exists a vicinity $U$ of $P^0$ where
$P^0+\mathbf{Q}(P^0,P^*)$ coincides with $\Phi(P^0)$:
$$(P^0+\mathbf{Q}(P^0,P^*))\cap U=\Phi(P^0)\cap U\, .$$
\end{proposition}
\begin{proof}
The inclusion $(P^0+\mathbf{Q}(P^0,P^*))\cap U\subset \Phi(P^0)\cap U$ is proven in the
second part of the proof of Proposition~\ref{OrderJustify1} because $\mathcal{K}_J
\subset \mathcal{K}$. We have to prove the inclusion $(P^0+\mathbf{Q}(P^0,P^*))\cap
U\supset \Phi(P^0)\cap U$.

Let us use the combinatorial description of compartments  and cones
(\ref{CompartConeDescr}). We assume that  $P^0\in \mathcal{C}_{\sigma}$ for a surjection
$\sigma :\{1,2,\ldots ,n\} \to \{1,2,\ldots ,k+1\}$. Let us recall that $k=\dim
\mathcal{C}_{\sigma}$. If $k=n-1$ then $\mathcal{C}_{\sigma}$ is an open subset of the
distribution space and the preimage of every $l=1,2,\ldots, n$ consists of one point. For
every $P\in \mathcal{C}_{\sigma}$ the cone $\mathbf{Q}(P,P^*)$ coincides with
$\mathbf{Q}(P^0,P^*)$ and due to Proposition~\ref{OrderJustify1} there exists a vicinity
$U$ of $P^0$ where $P^0+\mathbf{Q}(P^0,P^*)$ coincides with $\Phi(P^0)$.

Let $k<n-1$. Then for some $i=1, \ldots , k+1$ the preimage of $i$ includes more than 1
point, $|\sigma^{-1}(i)|>1$. Let $I$ be the set of such $i$ and $S_i=\sigma^{-1}(i)$ is
the preimage of $i$. Due to (\ref{CompartConeDescr}),
$$\mathbf{Q}(P^0,P^*)={\rm cone}\{\gamma^{ij} \ | \ \sigma(j)=\sigma(i)+1 \}.$$
For a sufficiently small ball $U_r$ with the centre $P^0$ and  $P\in
(P^0+\mathbf{Q}(P^0,P^*))\cap U$ the cone $\mathbf{Q}(P,P^*)$ may include also some
$\gamma^{ij}$ with $\sigma (i)=\sigma(j)$ but
\begin{equation}\label{ConeExtencion}
\mathbf{Q}(P,P^*)\subset{\rm cone}\{\gamma^{ij} \ | \ \sigma(j)=\sigma(i)+1\mbox{ or }
\sigma(j)=\sigma(i) \}.
\end{equation}

Let us prove that for any Markov chain with equilibrium $P^*$ for sufficiently small time
$\tau>0$ and a ball  $U_{r/2}$ with the centre $P^0$ the solutions of the Kolmogorov
equations $P(t)$ do not leave $P^0+\mathbf{Q}(P^0,P^*)$ during the time interval
$[0,\tau]$ if $P(0) \in (P^0+\mathbf{Q}(P^0,P^*))\cap U_{r/2}$.

A set $V$ is positively invariant with respect to a dynamical system if every motion that
starts in $V$ at $t=0$ remains there for $t>0$. Let a convex set $V$ be positively
invariant with respect to several dynamical system given by Lipschitz vector fields
$\mathbf{w}_1, \ldots, \mathbf{w}_r$. Then $V$ is positively invariant with respect to
any combination $\mathbf{w}=\sum_j f_j \mathbf{w}_j$, where $f_j$ are non-negative
functions and $\mathbf{w}$ is a Lipschitz vector field. Therefore, the problem of
positive invariance of a convex set with respect to such combinations of vector fields
can be ``split'' into problems of the positive invariance of $V$ with respect to summands
$\mathbf{w}_j$. Due to the second decomposition theorem, we can always assume that the
vector field of the Kolmogorov equation for the Markov kinetics is a linear combination
of the vector fields of the pairs of elementary transitions $A_i\rightleftharpoons A_j$
with the same equilibrium. The coefficients in these combinations are non-negative
functions.

The motion $P(t)$ with $P(0)\in (P^0+\mathbf{Q}(P^0,P^*))$ does  not leave
$(P^0+\mathbf{Q}(P^0,P^*))$ in time $t \in [0,\tau]$ if $\D P(t)/\D t \in
\mathbf{Q}(P^0,P^*)$ on $[0,\tau]$.

The cone $\mathbf{Q}(P^0,P^*)$ is generated by vectors $\gamma^{ij}$ with
$\sigma(j)=\sigma(i)+1$. To generate a cone $\mathbf{Q}(P,P^*)$ for a point $P\in U_r$ we
have to add to the set of  $\gamma^{ij}$ ($\sigma(j)=\sigma(i)+1$) some of $\gamma^{ij}$
with $\sigma(j)=\sigma(i)$. Let us consider the pyramid (compare to (\ref{pyramid}))
$$\mathcal{Q}(P^0)=\left\{\sum_{\sigma(j)=\sigma(i)+1} \theta_{ij}\gamma^{ji}\, \Big| \, \theta_{ij}\geq 0, \,
\sum_{(i,j)\in J} \theta_{ij} < 1  \right\}\, .$$ We will prove that the set $P^0+a
\mathcal{Q}(P^0)$ is positively invariant with respect to any first order kinetics with
transitions $A_i \rightleftharpoons A_j$ ($i,j\in S_l$) and equilibrium $P^*$ for any
$l=1,\ldots, k+1$.

It is sufficient to consider dynamics in projections on the coordinate  subspace
$\mathbf{R}_{S_l}$ with coordinates $p_i$, $i\in S_l$ for every $l\in I$ separately. In
this space, vectors $\gamma^{ij}$ ($i,j\in S_l$) correspond to the standard first order
kinetics like (\ref{QuasiChemKol}) with the reduced vector $P\in \mathbf{R}_{S_l}$ but
without compulsory unit balance ($\sum_{i\in S_l }p_i =const$ with any $const>0$). A
projection of $P^0$ on $\mathbf{R}_{S_l}$, $P^0_{S_l}$ is an equilibrium for this first
order kinetics with the balance $\sum_{i\in S_l }p_i =\sum_{i\in S_l }p^0_i$ because
$\frac{p^0_i}{p^*_i}=\frac{p^0_j}{p^*_j}$ for $i,j\in S_l$.

The vectors $\gamma^{ij}$ that generate $\mathbf{Q}(P^0,P^*)$ ($\sigma(j)=\sigma(i)+1$)
(\ref{ConeExtencion}) have non-zero projections on $\mathbf{R}_{S_l}$  if and only if
either $l=\sigma(j)=\sigma(i)+1$ or $\sigma(j)=\sigma(i)+1=l+1$. In the first case,
$l=\sigma(j)=\sigma(i)+1$, vector $\gamma^{ij}$ is the standard basis vector $e_j$ in
$\mathbf{R}_{S_l}$. In the second case, $\sigma(j)=\sigma(i)+1=l+1$, we have
$\gamma^{ij}=-e_i$. If $l=1$ then only the second case is possible, and  if $l=k+1$ then
only the first case can take place.

Let $V_l=\{P\in \mathbf{R}_{S_l}\ | \ p_i\geq 0, \sum_{i\in S_l}p_i<1\}$. The projection
of the pyramid $\mathcal{Q}(P^0)$ onto $\mathbf{R}_{S_l}$ is $\mbox{conv}(V_l-V_l))$ if
$1<l<k+1$; it is $V_l$ if $l=k+1$ and $-V_l$ if $l=1$. (For sets $X,Y$, the sum $X+Y$ is
the set of all sums $x+y$ ($x\in X$, $y\in Y$),  the difference $X-Y$ is the set of all
differences $x-y$, therefore $V-V$ is not $\{0\}$ if $V$ includes more than one element.)

The set $V_l$ is positively invariant with respect to the first order kinetics in
$\mathbf{R}_{S_l}$. Therefore, the following sets are also positively invariant with
respect to the first order kinetics in $\mathbf{R}_{S_l}$ with equilibrium $P^0_{S_l}$
for every $a>0$:
$$P^0_{S_l}+a V_l, \,(P^0_{S_l}+a V_l)-a V_l,\, \mbox{conv}((P^0_{S_l}+a V_l)-a V_l)\, .$$
Thus, the set $P^0+a \mathcal{Q}(P^0)$ is positively invariant with respect to any first
order kinetics with transitions $A_i \rightleftharpoons A_j$ ($i,j\in S_l$) and
equilibrium $P^*$ for any $l=1,\ldots, k+1$. A combination of  these statements for all
$l=1, \ldots , k+1$ finalizes the proof.
\end{proof}

This proposition finalizes the justification of the use of the cone of the tangent
directions $\mathbf{Q}(P^0,P^*)$ in the definition of the local minimum of the Markov
order (\ref{localExtr}).

\section{Equivalence of the maxima of all entropies and the Markov order approaches
\label{Sec:Equival}}

The cone $\mathbf{Q}(P^0,P^*)$ is a piecewise constant function of $P^0$: it is the same
for all $P^0$ from one compartment $ \mathcal{C}_{\sigma}$ and, hence, depends on
$\sigma$  only. Therefore, if the condition of the local minimum (\ref{localExtr}) holds
for one $P^0 \in L \cap \mathcal{C}_{\sigma}$ then it holds also for all elements of $L
\cap \mathcal{C}_{\sigma}$. There is a finite number of compartments
$\mathcal{C}_{\sigma}$.

Let the linear manifold of conditions $L$ be given by the values of moments $\sum_i
m_{ri}p_i=M_r$, $L^0= \ker m$ and $L\cap \Delta^{n-1}_+ \neq \emptyset$. The set of all
conditional local minima of the Markov order on the linear manifold of conditions $L$ is
\begin{equation}\label{CombinLocMin}
\boxed{\bigcup \left\{L \cap \mathcal{C}_{\sigma} \ \big| \ L \cap \mathcal{C}_{\sigma}
\neq \emptyset \; \mbox{ and }\; L^0\cap Q_{\sigma}=\{0\}\right\}\, ,}
\end{equation}
where $\mathcal{C}_{\sigma}$ and $Q_{\sigma}$ are defined by (\ref{CompartConeDescr}). It
is sufficient to find all $\sigma$ such that $L \cap \mathcal{C}_{\sigma} \neq \emptyset$
and $L^0\cap Q_{\sigma}=\{0\}$ and then describe the union of the compartments
$\mathcal{C}_{\sigma}$ for these $\sigma$.

The approach based on the minimization of all $f$-divergencies seems to be very
different. For all monotonically increasing functions $g$ we have to solve the equations
for the Lagrange multipliers and represent the probability distribution in the form
(\ref{LagrangeMaxEntAnsw}). Nevertheless, these approaches are equivalent and describe
the same set of the ``conditionally maximally disordered distributions''.

\begin{theorem}\label{EquivalenceOrders}A positive distribution $P^0\in L$ satisfies the local conditional minimum conditions
of the Markov order (\ref{localExtr}) if and only if there exists a strictly monotonic
function $g$ on $\mathbb{R}$ with ${\rm im}\, g=(0,\infty)$ such that the conditions
(\ref{LagrangeMaxEntAnsw}) hold for some Lagrange multipliers and for $p_i^0=p_i$.
\end{theorem}

This means that every conditionally minimal distribution of the Markov order on the
linear manifold $L\cap \Delta^{n-1}_+$ is a conditional minimum on $L\cap \Delta^{n-1}_+$
of a strictly convex $f$-divergence (\ref{Morimoto}).
\begin{proof}
Due to the classical theorems about separation of convex sets and linear spaces by linear
functionals \cite{Rockafellar1970}, a distribution $P^0$ satisfies the condition of the
local minimum (\ref{localExtr}) if and only if there exists a linear functional
$\psi(P)=\sum_i \psi_i p_i$ such that $\psi|_L=\psi(P^0)=const$ and $\psi(P)>\psi(P^0)$
for every $P\in P^0 + \mathbf{Q}(P^0,P^*)$ if $P\neq P^0$. In other words, $\psi|_{L^0}
\equiv 0$ and
\begin{equation}\label{psiPositive}
(\psi_i-\psi_j)\left(\frac{p_i^0}{p_i^*}-\frac{p_j^0}{p_j^*} \right)>0\; \mbox{ if } \;
\frac{p_i^0}{p_i^*}\neq \frac{p_j^0}{p_j^*}\, .
\end{equation}
according to the definition of $\mathbf{Q}(P,P^*)$ (\ref{COneQPP}). Condition
$\psi|_{L^0} \equiv 0$ is equivalent to the existence of the coefficients $\lambda_r$
such that for all $i$
$$\psi_i = \sum_r \lambda_r m_{ri} \, .$$
Condition (\ref{psiPositive}) is equivalent to the existence of a strictly monotonic
function $\eta(x)$ defined for $x\geq 0$ such that $$\psi_i =
\eta\left(\frac{p_i^0}{p_i^*}\right) \, .$$ To find such a function $\eta(x)$ we can take
the known values $\psi_i$ for $x={p_i^0}/{p_i^*}$ and then use, for example, linear
interpolation $\eta(x)$ between ${p_i^0}/{p_i^*}$. To extrapolate $\eta(x)$ from
$\max\{{p_i^0}/{p_i^*}\}$ to $+\infty$ we can use an increasing linear function. To
extrapolate $\eta(x)$ on the interval $(0,\min\{{p_i^0}/{p_i^*}\})$ we can use
$\varepsilon \log x + const$.

Finally, we can take $h'(x)=\eta(x)$, $h(x)=\int \eta (\xi)\, \D \xi$; and $g(y)$ is the
inverse function: $g(\eta(x))=x$ for $x\geq 0$. The distribution $P^0$ is the local
minimum of $H_h(P\|P^*)$ on $L$.

Conversely, if $P^0$ is a minimum of a strictly convex Lyapunov function $H$ on $L$ and
$\D H /\D t|_{P^0}<0$ for every Markov chain with equilibrium $P^*$ for which $P^0$ is a
non-equilibrium distribution then we can take $$\psi_i =-\left.\frac{\partial H}{\partial
p_i}\right|_{P^0}\, .$$
 This choice of $\psi_i$ provides (\ref{psiPositive}) (because $H$ is strictly decreasing
 in time Lyapunov function) and $\psi|_{L^0} \equiv 0$ because grad$H$ is orthogonal to $L$ (the
 condition of local minimum).
\end{proof}

This equivalence of two definitions of the maximally uncertain distribution under given
conditions has several important consequences.

Let us introduce the notion of the (global) {\em Markov order} \cite{GorGorJudge2010}.
\begin{itemize}
\item{If for distributions $P^0$ and $P^1$ there exists such a Markov process with
    equilibrium $P^*$ that for the solution of the Kolmogorov equation with
    $P(0)=P^0$ we have $P(1)=P^1$ then we say that $P^0$ and $P^1$ are connected by
    the Markov preorder \cite{GorGorJudge2010} with equilibrium $P^*$ and use
    notation $P^0 \succ^0_{P^*} P^1$.}
\item{The (global) { Markov order} is the closed transitive closure of the Markov
    preorder. For the Markov order with equilibrium $P^*$ we use notation $P^0
    \succ_{P^*} P^1$.}
\end{itemize}

The {\em local Markov order} at point $P^0$ is just a vector order generated by the
tangent cone $\mathbf{Q}(P^0,P^*)$ \cite{GorGorJudge2010}. We use for this local order
the notation $>_{P^0,P*}$: $$P>_{P^0,P^*} P' \; \mbox{ if }\; P'-P \in
\mathbf{Q}(P^0,P^*).$$ The proofs of Propositions~\ref{OrderJustify1} and
\ref{OrderJustify2} give us the possibility to use  the relation $P^0>_{P^0,P^*} P^1$
instead of the Markov preorder for the definition of the Markov order minimizers on
linear manifolds. The relation $P^0>_{P^0,P^*} P^1$ is defined by the local Markov order
in a vicinity of $P^0$:
$$P^1-P^0\in \mathbf{Q}(P^0,P^*).$$ The cone $\mathbf{Q}(P^0,P^*)$ depends on $P^0$,
therefore, the relation $P^0>_{P^0,P^*} P^1$ is antisymmetric locally, in a vicinity of
$P^0$.

\begin{remark}It is possible to generate the Markov order by the relation $P^0>_{P^0,P^*} P^1$.
Let us specify the vicinity of $P^0$ where this relation is defined and introduce a new
relation:
 $P^0>_{P^*}^0 P^1$ if $P^0>_{P^0,P^*} P^1$  for all $i,j=1, \ldots, n$ and
$$\left(\frac{p^0_i}{p^*_i}-\frac{p^0_j}{p^*_j}\right)\left(\frac{p^1_i}{p^*_i}-\frac{p^1_j}{p^*_j}\right)\geq
0\, .$$  This condition means that the pairs of numbers
$(\frac{p^0_i}{p^*_i},\frac{p^0_j}{p^*_j})$ and
$(\frac{p^1_i}{p^*_i},\frac{p^1_j}{p^*_j})$ cannot have an opposite order on the real
line. The closed transitive closure of the relation $P^0>_{P^*}^0
 P^1$ is the Markov order $P^0 \succ_{P^*} P^1$
\end{remark}
Let $L$ be a linear manifold in the space of distributions. By definition, $P^0\in L$ is
a minimal point on $L \cap \Delta^{n-1}_+$ with respect to the order $ \succ_{P^*}$ if
and only if there is no point $P^1\in L\cap \Delta^{n-1}_+ $, $P^1\neq P^0$ such that
$P^0 \succ_{P^*} P^1$.
\begin{corollary}\label{GlobalMinimum}
$P^0\in L\cap \Delta^{n-1}_+$ is a minimal point on $L \cap \Delta^{n-1}_+$ with respect
to the (global) Markov order if and only if it satisfies the local minimum condition
(\ref{localExtr}).
\end{corollary}
\begin{proof}
If $P^0\in L \cap \Delta^{n-1}_+$ is a minimal point on $L\cap \Delta^{n-1}_+ $ with
respect to the (global) Markov order then it satisfies the condition (\ref{localExtr})
due to the definition of the Markov order through the transitive closure of the relation
$P^0 \succ^0_{P^*} P^1$ and Propositions~\ref{OrderJustify1} and \ref{OrderJustify2}.

Let $P^0$ satisfy the local minimum condition (\ref{localExtr}). Then there exists a
divergence $H_h(P\| P^*)$ with strictly convex $h(x)$ ($x\geq 0$) such that $P^0$ is a
local minimum of $H_h(P\| P^*)$ on $L$. Because of strong convexity, this local minimum
is a global one. $H_h(P\| P^*)$ is a Lyapunov function for all Markov chains with
equilibrium $P^*$. Therefore, a broken line, which is combined from solutions of the
Kolmogorov equations for such Markov chains and starts at $P^0$, leaves a small vicinity
of $P^0$ (Propositions~\ref{OrderJustify1} and \ref{OrderJustify2}) and never returns in
a sufficiently small vicinity of $L$. Thus, for the closed transitive closure of the
relation $P \succ^0_{P^*} P'$, point $P^0$ is a minimal point on $L$.
\end{proof}
Of course, there may be infinitely  many minimal points of the Markov order on $L$ and
each of them corresponds to a different Lyapunov functions $H_h(P\| P^*)$.

Another remarkable order on the space of distributions is $P^0 > _{H,P^*} P^1$ if for all
strictly convex functions $h(x)$ ($x\geq 0$)
$$H_h(P^0\|P^*) > H_h(P^1\|P^*)\, ,$$
that is, $P^1$ is closer to equilibrium than $P^0$  with respect to {\em all}
divergencies $H_h(P\| P^*)$.

\begin{corollary}\label{CorLocalMIn}
For any linear manifold $L$ in the distribution space the minimal elements of the Markov
order $\succ_{P^*}$ on $L \cap \Delta^{n-1}_+$ coincide with the minimal elements of the
order $> _{H,P^*}$ on $L \cap \Delta^{n-1}_+ $.
\end{corollary}
\begin{proof}We just have to combine Theorem~\ref{EquivalenceOrders} with
Corollary~\ref{GlobalMinimum}.
\end{proof}
Thus, the minimal elements of the orders $\succ_{P^*}$ and $> _{H,P^*}$ on the linear
manifolds coincide. Nevertheless, it is necessary to mention the difference between these
orders. Let $P^0$ be a distribution. For $> _{H,P^*}$ the set of distributions $\{P \ | \
P^0 > _{H,P^*} P\}$ is convex as an intersection of convex sets $\{P \ | \ H_h(P^0\|P^*)>
H_h(P\|P^*)\}$ for various strictly convex $h$. This is not the case for the Markov
order. The set of distributions $\{P \ | \ P^0 \succ_{P^*} P\}$ may be non-convex. The
examples may be extracted from the papers \cite{Gorban1979,Zylka1985} (see
Fig.~\ref{FigLowerCones}).

\begin{figure}
\centering{
\includegraphics[width=0.6\textwidth]{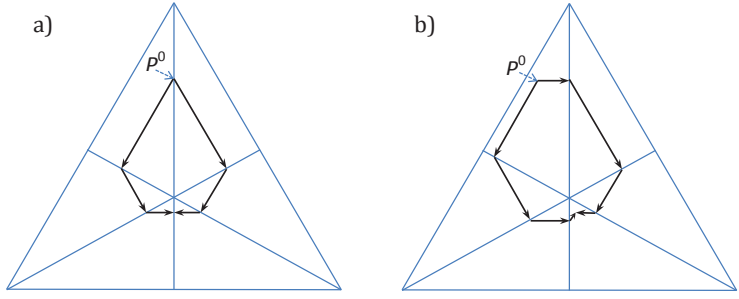}
\caption{\label{FigLowerCones}The set $\{P \ | \ P^0 > _{H,P^*}\}$ for different $P^0$ and for the Markov
chains with three states and equilibrium ($p_i^*=1/3$): (a) $\{P \ | \ P^0
> _{H,P^*}\}$ is convex, (b) it is not convex. The border of the set $\{P \ | \ P^0 > _{H,P^*}\}$ is
highlighted by bold lines. The arrows on these lines correspond to the directions of the extreme rays of the cones $\mathbf{Q}(P,P^*)$
(i.e. the angles represented in Fig.~\ref{3stateCones}).}}
\end{figure}

\begin{corollary}\label{WeakestLocal}
Let $P^0 \succ_{P^*} P^1$. Then $P^1 \in P^0+\mathbf{Q}(P^0,P^*)$.
\end{corollary}
\begin{proof}
Let us apply Corollary~\ref{GlobalMinimum} to all support hyperplanes $L$ of the convex
set $(P^0+\mathbf{Q}(P^0,P^*))$ for which $(P^0+\mathbf{Q}(P^0,P^*))\cap L=\{P^0\}$.
\end{proof}

\section{Example: generalization of the normal distribution}

In this section, we discuss distributions $p(x)$ on a continuous space of states, the
non-negative real semi-axis, $\mathbb{R}_+=\{x \ | \ x\geq 0\}$. We have in mind two
classical examples of distributions of the quantity bounded from below: energy (physics)
and wealth (economics and microeconomics).

Let two moments be fixed, the total probability $M_0=\int_0^{\infty}p(x) \, \D x$ and the
average quantity $M_1=\int_0^{\infty}x p(x) \, \D x$. The conditional maximization of the
classical Boltzmann--Gibbs--Shannon entropy gives:
$$\int p(x) (\ln p(x)-1) \, \D x \to \min \;\mbox{ for given }\; \int_0^{\infty}p(x) = M_0, \; \int_0^{\infty}p(x) \, \D
x=M_1\, ;$$
$$\ln p(x)=\lambda_0 + \lambda_1 x, \; p(x)=\exp (\lambda_0 + \lambda_1 x),\; \exp \lambda_0 = -\lambda_1 \; \exp \lambda_0=M_1 \lambda_1^2\, ;$$
\begin{equation}\label{BoltzmannDist}
\boxed{p^*(x)=\frac{1}{M_1} \exp\left(-\frac{x}{M_1}\right)\, . }
\end{equation}

This Boltzmann distribution appears always as a first candidate for the equilibrium
distribution of an additive conserved quantity bounded from below. Khinchin (1943)
clearly explained this law as a version of the limit theorem \cite{Khin}.

Technically, it is not difficult to involve the higher moments and obtain the
distribution of the form
\begin{equation}\label{manymoments}
p(x)=\exp (\lambda_0 + \lambda_1 x+ \lambda_2 x^2 +\ldots + \lambda_r x^r)\, .
\end{equation}
One can expect that this extension of the set of moments may improve the description.
This is a traditional belief in Extended Irreversible Thermodynamics (EIT)
\cite{JouEIT2001}.

There may be many different approaches to evaluation of the quality of the approximation
(\ref{manymoments}) but at least one important property of these functions is wrong: the
asymptotic behavior at large $x$ is $p(x) \asymp \exp (-const\times x^r)$. These
``super-light'' tails  of the distribution $p(x)$ change qualitatively with the change of
the order $r$ in (\ref{manymoments}).

If we use, for example, the ``regularizing'' forth moment in the moment chain for the
Boltzmann equation \cite{Levermore1996} then we corrupt the ${\rm e}^{-const\times v^2}$
tails of the Maxwell distribution. Therefore, other approaches which do not modify the
tails of the distribution qualitatively (like \cite{GorKar2006}) may be more appreciated.

The  asymptotic behavior of the distribution's tails was thoroughly studied in many
cases. Very often, the tails of the distributions are, without a doubt, heavier than
normal ${\rm e}^{-const\times x^2}$ and definitely are not cut as ${\rm e}^{-const\times
x^4}$. For example, it is demonstrated that the distribution of money between people has
the exponential tail with a possible transformation into a heavier power tail for very
rich people \cite{Yakovenko2009}.

The general solution (\ref{LagrangeMaxEntAnsw}) with the Boltzmann equilibrium
(\ref{BoltzmannDist}) gives the following expression  instead of (\ref{manymoments})
$$p(x)=g(\lambda_0 + \lambda_1 x+ \lambda_2 x^2 +\ldots + \lambda_r x^r) \frac{1}{M_1}
\exp\left(-\frac{x}{M_1}\right)\, , $$ where $g$ is a monotonically increasing  function. In particular,
for the moments $M_0$, $M_1$ and $M_2$ we obtain
\begin{equation}\label{gennorm}
p(x)=g(\lambda_0 + \lambda_1 x+ \lambda_2 x^2) \frac{1}{M_1} \exp\left(-\frac{x}{M_1}\right)\, .
\end{equation}

There are four qualitatively different cases of (\ref{gennorm}). Let $\lambda_2 \neq 0$
and $\mu=-\frac{\lambda_1}{2\lambda_2}$. Then
\begin{equation}\label{genNorm1}
\boxed{p(x)= \frac{f(x)}{M_1} \exp\left(-\frac{x}{M_1}\right)\, ,}
\end{equation}
 and
\begin{enumerate}
\item{if $\mu \leq 0$ and $\lambda_2>0$ then $f(x)$ is a monotonically increasing
    function on $[0,\infty)$;}
\item{if $\mu \leq 0$ and $\lambda_2<0$ then  $f(x)$ is a monotonically decreasing
    function  on $[0,\infty)$;}
\item{if $\mu > 0$ and $\lambda_2>0$ then  $f(x)$ is a monotonically increasing
    function on $[\mu,\infty)$ and $f(x)=f(2\mu-x)$ for $x \in[0,\mu]$;}
\item{if $\mu > 0$ and $\lambda_2<0$ then  $f(x)$ is a monotonically decreasing
    function  on $[\mu,\infty)$ and $f(x)=f(2\mu-x)$ for $x\in [0,\mu]$.}
\end{enumerate}
Each of these ``generalized normal distributions'' (\ref{genNorm1}) is a minimizer of the
corresponding $f$-divergence. For the construction of such a divergence in general case,
it is convenient to define the convex functions $h$ in (\ref{Morimoto}) with values on an
extended real line with additional possible value $+\infty$. This is a natural general
definition of convex functions \cite{Rockafellar1970}. In case 1 ($\mu\leq 0$,
$\lambda_2>0$, and $f$ increases), we can take in (\ref{gennorm}), (\ref{genNorm1}) without loss of generality
$\mu=0$, $f(x)=g(x^2)$, and $g(y)=f(\sqrt{y})$. The monotonically increasing function $g(y)$ is,
therefore, defined on $[0,\infty)$ with the set of values $[\underline{g},\overline{g})$,
where $\underline{g}=f(0)\geq 0$, $\overline{g}=\lim_{x\to \infty} f(x) > 0$ and the
upper limit may be finite or infinite. The inverse function $\xi(z)$ is defined for $z\in
[\underline{g},\overline{g})$ with the interval of values $[0,\infty)$. Let us take
\begin{equation}
h'(z)=\left\{\begin{array}{ll}
0 &\mbox{ if } z< \underline{g}; \\
\xi(z) &\mbox{ if }z\in [\underline{g},\overline{g}); \\
\infty &\mbox{ if } z \geq \overline{g};
\end{array}\right. \;\;
h(z)=\left\{\begin{array}{ll}
0 &\mbox{ if } z< \underline{g}; \\
\int_{\underline{g}}^z \xi(\varsigma) \D \varsigma &\mbox{ if }z\in [\underline{g},\overline{g}]; \\
\infty &\mbox{ if } z > \overline{g}.
\end{array}\right.
\end{equation}
The improper integral  $\int_{\underline{g}}^{\overline{g}} \xi(\varsigma) \D \varsigma$
may take finite or infinite values.

Similarly, in case 2 ($\mu\leq 0$ and $\lambda_2<0$, $f$ decreases) we define
$g(y)=f(\sqrt{-y})$ for $y\in (-\infty,0]$. The function $g(y)$ monotonically increases
and takes values on $(\underline{g},\overline{g}]$, where $\underline{g}=\lim_{x\to
\infty} f(x)$ and $\overline{g}=f(0)$. The inverse function $\xi(z)$ is defined for $z\in
(\underline{g},\overline{g}]$ with the interval of values $(-\infty, 0]$. In this case,
we can take
\begin{equation}
h'(z)=\left\{\begin{array}{ll}
-\infty &\mbox{ if } z < \underline{g}; \\
\xi(z) &\mbox{ if }z\in (\underline{g},\overline{g}]; \\
0 &\mbox{ if } z > \overline{g};
\end{array}\right. \;\;
h(z)=\left\{\begin{array}{ll}
\infty &\mbox{ if } z \leq \underline{g}; \\
-\int^{\overline{g}}_z \xi(\varsigma) \D \varsigma &\mbox{ if }z\in [\underline{g},\overline{g}]; \\
0 &\mbox{ if } z > \overline{g}.
\end{array}\right.
\end{equation}

In case 3, the construction is almost the same as for the case 1 but $f(x)=g((x-\mu)^2)$
and $g(y)=f(\sqrt{y}+\mu)$. In this case, $g(y)$ is a monotonically increasing function
defined on the interval $[0,\infty)$ with the set of values
$[\underline{g},\overline{g})$, where $\underline{g}=f(\mu)$ and $\overline{g}=\lim_{x\to
\infty} f(x)$. Similarly, for case 4, the construction of $h(z)$ is almost the same as in
case 2.

Thus, for every distribution in the form (\ref{genNorm1}) we can find a $f$-divergence
$H_h(P\|P^*)$, which conditional minimization produces this distribution. For example, if
in (\ref{genNorm1}) $f(x)=a x^{\beta}$ then we can take $h$ in the form
$h(z)=\frac{\beta}{\beta+ 2}(z/a)^{1+2/\beta}$.

\section{Conclusion}

{\em The Maxallent approach} aims to bring some order to the modern anarchy of the
measures of disorder. If there is no clear idea which entropy is better then we have to
use all of them together.

{\em The Markov order approach} was also proposed as an alternative to the entropic
anarchism. It is based on the idea that the disorder has to increase in random processes
with given equilibrium distribution, which is considered as the maximally disordered
state. Here, we have proved that these two approaches produce the same conditional
minimizers on the planes of given values of moments (Theorem~\ref{EquivalenceOrders}).

In this paper, we have considered several  relations between positive distributions:
\begin{enumerate}
\item{$P^0 \succ^0_{P^*} P^1$ if there exists a Markov chain with equilibrium $P^*$
    such that for the solution of the Kolmogorov equation $P(t)$  with $P(0)=P^0$ we
    have $P(1)=P^1$;}
\item{$P^0 \succ_{P^*} P^1$ if there exist integrable bounded functions $q_{ij}(t)$
    ($i,j=1,\ldots, n$, $i\neq j$, $t\geq 0$) such that $q_{ij}(t)$  satisfy the
    balance condition (\ref{MasterEquilibrium}) for given  $P^*$ ($p^*_i>0$) (for all
    $t\geq 0$), and $P(1)=P^1$ for solution $P(t)$ of the equations
    $$\frac{\D p_i}{\D t}= \sum_{j, \, j\neq i} (q_{ij}(t)p_j-q_{ji}(t)p_i) \;\; (i=1,\ldots, n)$$
    with $P(0)=P^0$ (that is, $\succ_{P^*}$ is the transitive closure of
$\succ^0_{P^*}$);}
\item{$P^0 >_{H,P^*} P^1$ if $H_h(P^0\|P^*) > H_h(P^1\|P^*)$ for all strictly convex
    functions $h(x)$ on a semi-axis $x\geq 0$.}
\item{$P^0>_{P^0,P^*} P^1$ if $P^1-P^0 \in \mathbf{Q}(P^0,P^*)$, where
    $\mathbf{Q}(P^0,P^*)$ is the cone of possible velocities $\D P/ \D t $
    (\ref{COneQPP}) at point $P^0$ for all Markov chains with equilibrium $P^*$.}
\end{enumerate}
All these relations are different. Three of them are antisymmetric, and one,
$P^0>_{P^0,P^*} P^1$, is locally antisymmetric, in a vicinity of $P^0$. Their
interrelations are described by the follows implications:
$$   (P^0 \succ^0_{P^*} P^1) \Rightarrow (P^0 \succ_{P^*} P^1)
\Rightarrow (P^0 > _{H,P^*} P^1 ) \Rightarrow (P^0>_{P^0,P^*} P^1)\, .$$

The local Markov order $P^0 >_{P^*} P^1$ is the weakest and the connection by a solution
of the Kolmogorov equation $P^0 \succ^0_{P^*} P^1$ is the strongest of these relations.
Nevertheless, locally, in a small vicinity of a positive non-equilibrium distribution
$P^0$, these relations coincide and they define the same set of locally minimal
distributions on a linear manifold of conditions $L$ (Propositions~\ref{OrderJustify1},
\ref{OrderJustify2}, Theorem~\ref{EquivalenceOrders}, Corollaries~\ref{GlobalMinimum},
\ref{CorLocalMIn} and \ref{WeakestLocal}).

Of course, there is  the other, the classical way to reduce the variability of the
measures of disorder. The divergences $H(P\| P^*)$ can be defined by their main
properties. This is an axiomatic approach: we postulate some ``natural properties'' of
the divergence, then find the divergences with these properties, evaluate the result and
decide whether we have to change the system of axiom or not. The axiomatic approach to
definition of entropy was used by Shannon \cite{Shannon1948} and elaborated in detail by
Khinchin \cite{Khinchin1957}.

Two distinguished additivity properties are important for the Maxent reasoning:
\begin{itemize}
\item{Additivity on the algebra of states: $H(P\| P^*)$ is a sum in states $$H(P\| P^*)=\sum_i \eta( p_i,p^*_i)\, .$$}
\item{Additivity with respect to the joining of independent subsystems. This means
    that if $P$ and $P^*$ are products of distributions then $H(P\| P^*)$ is the sum
    of the corresponding entropies: if $P=(p_{jl})=(q_j r_l)$ and
    $P^*=(p^*_{jl})=(q^*_j r^*_l)$ then $H(P\| P^*)=H(Q\| Q^*)+H(R\| R^*)$.}
\end{itemize}

If we join the first additivity property with the requirement that the divergence should
be a Lyapunov function for all Markov chains with equilibrium $P^*$ then we get $H_h(P\|
P^*)$ of the form (\ref{Morimoto}) \cite{ENTR3,Amari2009,GorGorJudge2010}. If we add the
second additivity property and require continuity of $H_h(P\| P^*)$ for all values of $P$
(including vectors with some $p_i=0$) then the classical Boltzmann--Gibbs--Shannon
relative entropy will be the only possibility (that is, $H_h(P\| P^*)$ with $h(x)=x\ln x$
up to unimportant constant factors and summand). If we relax the requirement of the
continuity to the set of strictly positive distributions then we will get the
one-parametric family $H_h(P\| P^*)$ with $h(x)=\beta  x\ln x -(1-\beta)\ln x$
\cite{ENTR3,GorGorJudge2010}.

Let us accept the point of view that the divergency is an order. Then the values are not
important and all the divergencies connected by a monotonic transformation of a scale,
$H=f(H')$ (with a monotonically increasing $f$), are equivalent. If the first additivity
property is valid in one scale, and the second may be valid in another one, then one more
one-parametric family appear, the Cressie--Read divergences (see Appendix~A)
\cite{ENTR3,GorGorJudge2010}. The Tsallis entropy is a particular case of them. The
Boltzmann--Gibbs--Shannon relative entropy (or the Kullback--Leibler entropy, which is
the same), the convex combination of $H_h(P\| P^*)$ and $H_h(P^*\| P)$ for $h(x)=x\ln x$,
and the Cressie--Read divergences (including the Tsallis relative entropy) form the
``entropic aristocracy'' distinguished mostly by the additivity properties.

If we accept the additivity on the algebra of states (i.e., the trace form) and the
additivity with respect to joining of independent subsystems, both, then we have to use
some of these functions. If additivity with respect to joining of independent subsystems
seems to be too restrictive then we have to take the wider class of divergencies, for
example, $H_h(P\| P^*)$ of the form (\ref{Morimoto}). If we reject the requirement of the
trace form then the variety of the admissible divergences becomes even richer. This
uncertainty in the choice of divergence forces us to use the Maxallent approach.

The Maxallent approach produces a set of conditionally maximally disordered distributions
instead of a single distribution that maximizes a selected distinguished entropy in the
usual Maxent method. These Maxallent sets of distributions may be considered as
probabilistic analogues of the type-2 fuzzy sets introduced by L. Zadeh \cite{Zadeh1975}
to capture the uncertainty of the fuzzy systems. The Maxallent approach is invented to
manage the uncertainty of the measures of uncertainty. If there is no uncertainty of
uncertainty then the set of distributions reduces to a single distribution.

The decomposition theorems for Markov chains provide us with tools for the efficient
calculation of the Markov order. Following \cite{GorbanEq2012arX}, we compare the general
Markov chains and the reversible chains with detailed balance. For any general chain
there is a reversible chain with the same velocity vector at a given point. The classes
of general and reversible chains locally coincide because they have the same cone of
possible velocities at every non-equilibrium distribution (the second decomposition
theorem, Appendix~B). This theorem gives us the possibility to describe the set of the
conditionally maximally uncertain distributions combinatorially, in the finite form
(\ref{CombinLocMin}).

For the classical Boltzmann--Gibbs--Shannon entropy the distribution on $\mathbb{R}_+$
with two given moments has the Gaussian form $a \exp(-b(x-c)^2)$. The class of the
Maxallent distributions on $\mathbb{R}_+$ with two given moments is also simple
(\ref{genNorm1}) but much richer. It can be produced by multiplication of the Boltzmann
distribution (\ref{BoltzmannDist}) by a monotonic function or unimodal function (with one
local maximum) or by a function with one local minimum.

There exists an attractive possibility: if a distribution can be obtained in the
Maxallent approach then it is a conditional minimum of a divergence. If we find or guess
a distribution of the Maxallent type for an empirical system then we can restore the
divergence and then use it in the standard Maxent reasoning.

The Maxallent approach is, surprisingly, efficient enough to analyze some practical
problems.  It gives an answer that does not depend on the subjective choice and,
therefore, returns us to the ``mission'' of information theory: {\em ``to eliminate
the psychological factors involved...''} \cite{Hartley1928}. At the same time, it has a
solid basis in the theory of Lyapunov functions for the Kolmogorov equations.

Now, essential mathematical work on the basic notion of entropy is needed. Gromov
suggests that the natural mathematical language for this work will involve nonstandard
analysis and category theory \cite{Gromov2012}. These abstract languages seem to be
closer to the basic intuition than the set theory of Cantor and the $\varepsilon -\delta$
reasoning of the classical analysis. Nevertheless, the basic idea of Maxallent is so
simple and natural, that it should persist in the future advanced theory of entropy: {\em
order is something that decreases in Markov processes}.

\section*{Appendix A. The most popular examples of $H_h(P \| P^*)$}

The most popular examples of $H_h(P \| P^*)$ are \cite{GorGorJudge2010}:

\begin{enumerate}
\item{Let $h(x)$ be the step function, $h(x)=0$ if $x=0$ and $h(x)=-1$ if $x> 0$. In
    this case,
\begin{equation}\label{Hartley}
H_h(P \| P^*)=-\sum_{i, \ p_i>0} 1\, .
\end{equation}
The quantity  $-H_h$ is the number of non-zero probabilities $p_i$ and does not
depend on $P^*$. Sometimes it is called the Hartley entropy.}
\item{$h=|x-1|$, $$H_h(P \| P^*)=\sum_i |p_i - p_i ^*| \, ;$$ this is the
    $l_1$-distance between $P$ and $P^*$.}
\item{$h=x\ln x$,
\begin{equation}\label{Kullback--Leibler}
H_h(P \| P^*)=\sum_i p_i \ln \left(\frac{p_i}{p_i^*}
\right)=D_{\mathrm{KL}}(P\|P^*) \, ;
\end{equation}
 this is the usual Kullback--Leibler divergence or the relative
Boltzmann--Gibbs--Shannon (BGS) entropy;}
\item{$h=-\ln x$,
\begin{equation}\label{Burg}
H_h(P \| P^*)=-\sum_i p^*_i \ln \left(\frac{p_i}{p_i^*} \right) =
D_{\mathrm{KL}}(P^*\|P)\, ;
\end{equation}
this is the relative Burg entropy. It is obvious that this is again the
Kullback--Leibler divergence, but for another order of arguments. }
\item{Convex combinations of $h=x\ln x$ and $h=-\ln x$ also produces a remarkable
    family of divergences: $h=\beta x\ln x- (1-\beta) \ln x$ ($\beta \in [0,1]$),
\begin{equation}\label{ConvComb}
H_h(P \| P^*)=\beta D_{\mathrm{KL}}(P\|P^*) +
(1-\beta)D_{\mathrm{KL}}(P^*\|P)\, ;
\end{equation}
this convex combination of divergences was used by Gorban in the early 1980s
\cite{G11984} and studied further by Gorban and Karlin \cite{ENTR1}. It becomes a
symmetric functional of $(P, P^*)$ for $\beta=1/2$. There exists a special name for
this case, ``Jeffreys' entropy".}
\item{$h=\frac{(x-1)^2}{2}$,
\begin{equation}
H_h(P \| P^*)=\frac{1}{2}\sum_i \frac{(p_i-p_i^*)^2}{p_i^*}=H_2(P \| P^*)\, ;
\end{equation}
 this is the quadratic term in the Taylor expansion of the relative
Boltzmann--Gibbs-Shannon entropy, $D_{\mathrm{KL}}(P\|P^*)$, near equilibrium. We
have used its time derivative in (\ref{QuadLyap}).}
\item{$h=\frac{x(x^{\lambda}-1)}{\lambda (\lambda+1)}$,
\begin{equation}\label{Cressie--Read}
H_h(P \| P^*)=\frac{1}{\lambda (\lambda+1)}
 \sum_i p_i\left[ \left(\frac{p_i}{p_i^*} \right)^{\lambda} -1 \right]\, ;
\end{equation}
  this is the Cressie--Read (CR) family of power divergences \cite{CR1984} (the
  modern exposition of the history, properties and applications of these entropies is
  presented in \cite{CichockiAmari2010}). For this family we use the notation $H_{\rm
  CR \ \lambda}$. If $\lambda \to 0$ then
   $H_{\rm CR \ \lambda} \to  D_{\mathrm{KL}}(P\|P^*)$, this is the classical BGS
 relative entropy; if $\lambda \to -1$ then $H_{\rm CR \ \lambda} \to
 D_{\mathrm{KL}}(P^*\|P)$, this is the relative Burg entropy. }
\item{For the CR family in the limits $\lambda \to \pm \infty$ only the maximal terms
    ``survive". Exactly as we get the limit $l^{\infty}$ of $l^p$ norms for $p \to
    \infty$, we can use $({\lambda (\lambda+1)}H_{\rm CR \ \lambda})^{1/|\lambda|}$
    for $\lambda \to \pm \infty$ and write in these limits:
\begin{equation}\label{CR+infty}
 H_{\rm CR \ \infty}(P \| P^*) =\max_i\left\{\frac{p_i}{p_i^*}\right\}-1 \, ;
\end{equation}
\begin{equation}\label{CR-infty}
H_{\rm CR \ -\infty} (P \| P^*)= \max_i\left\{\frac{p_i^*}{p_i}\right\}-1\, .
\end{equation}
The existence of two limiting divergences $H_{{\rm CR \ \pm\infty}}$ seems very
natural: there may be two types of extremely non-equilibrium states: with a high
excess of current probability $p_i$ above $p_i^*$ and, inversely, with an extremely
small current probability $p_i$ with respect to $p_i^*$. }
\item{The Tsallis relative entropy \cite{Tsallis1988}:
    $h=\frac{(x^{\alpha}-x)}{\alpha - 1}$, $\alpha     >0$,
\begin{equation}\label{Tsallis}
H_h(P \| P^*)=\frac{1}{\alpha-1}\sum_i p_i\left[
\left(\frac{p_i}{p_i^*} \right)^{\alpha-1} -1 \right] \, .
\end{equation}
For this family we use notation $H_{\rm Ts \ \alpha}$.}
\end{enumerate}

\section*{Appendix B. The decomposition theorems}

\newtheorem*{firstdecompose}{The first decomposition theorem}

\begin{firstdecompose}Every Markov chain with a positive
equilibrium is a conic combination of simple cycles with the same equilibrium.
\end{firstdecompose}
\begin{proof}
If a non-zero Markov chain has a positive equilibrium then it cannot be acyclic: there
exists at least one oriented cycle of transitions with nonzero rate constants. The length
of this cycle can vary from 2 to $n$. The set of all Markov chains with a positive
equilibrium $P^*$ is an intersection of a linear subspace given by the balance equations
(\ref{MasterEquilibrium}) with the positive orthant $\mathbb{R}_+^{n(n-1)}$. This is a
polyhedral cone which does not include a whole straight line. It is well known in convex
geometry that every such polyhedral cone is a convex hull of a finite number of its {\em
extreme rays} \cite{Rockafellar1970}.  A ray $l$ with direction vector $x\neq 0$ is a set
$l=\{\kappa x\}$ ($\kappa \geq 0$). By definition, it is an extreme ray of a cone
$\mathbf{Q}$ if for any $u \in l$ and any $x,y \in \mathbf{Q}$, whenever $u = (x + y)/2$,
we must have $x,y\in l$.

Any extreme ray of the cone of Markov chains with equilibrium $P^*$ is a simple cycle
$A_{i_1} \to \ldots \to A_{i_k} \to A_{i_1}$ with rate constants $q_{i_{j+1} i_j}=\kappa
/p_j^*$. Indeed, let a non-zero Markov chain $Q$ with coefficients $q_{ij}$ belong to an
extreme ray of this cone. This chain includes a simple cycle with non-zero coefficients,
$A_{i_1} \to \ldots \to A_{i_k} \to A_{i_1}$ ($k\leq n$, all the numbers $i_1, \ldots,
i_k$ are different, $q_{i_{j+1}\,i_{j}}>0$ for $j=1,\ldots, k$, and $i_{k+1}=i_1$). For
sufficiently small $\kappa$ ($0<\kappa<\kappa_0$),
$q_{i_{j+1}\,i_{j}}-\frac{\kappa}{p^*_{i_{j}}}>0$ ($j=1,\ldots, k$). Let $Q_{\kappa}$ be
the same simple cycle with the rate constants $q_{i_{j+1} i_j}=\kappa /p_j^*$. Then for
$0<\kappa <\kappa_0$ vectors $Q\pm Q_{\kappa}$ also represent  Markov chains with the
equilibrium $P^*$. Obviously, $Q=\frac{(Q+Q_{\kappa})+(Q-Q_{\kappa})}{2}$, hence, $Q$
should be proportional to $Q_{\kappa}$, by the definition of extreme rays.

So, any Markov chain with a positive equilibrium $P^*$ is a linear combination with
positive coefficients of the cycles with the same equilibrium. This decomposition is
global, it does not depend on the current distribution $P$.
\end{proof}

\newtheorem*{seconddecompose}{The second decomposition theorem}
\begin{seconddecompose}
For every Markov chain with a positive equilibrium $P^*$ and any probability distribution
$P^0$ the vector $\D P/\D t|_{P^0}$ is a conic combination of the vectors $\D P/\D
t|_{P^0}$ for the  simple cycles of length two $A_i \rightleftharpoons A_j$ with the same
equilibrium.
\end{seconddecompose}
\begin{proof}
Let us start from a simple cycle $A_1 \to A_2 \to \ldots \to A_n \to A_1$  with the
constants $q_{i+1 i}=1 /p_i^*$, where $p_i^*>0$ is the equilibrium. At a non-equilibrium
distribution $P$ the right hand side of equation (\ref{MAsterEq0}) is the vector $\D P /
\D t = \mathbf{v}_n$ with coordinates
\begin{equation}\label{cycleVelocity}
\frac{\D p_j}{\D t}=(\mathbf{v}_n)_j=\frac{p_{j-1}}{p^*_{j-1}}
-\frac{p_{j}}{p^*_{j}}\, .
\end{equation}

The flux $A_{j}\to A_{j+1}$ is ${p_{j}}/{p^*_{j}}$. Let us find $A_j$ with the minimum
value of this flux and, for convenience, let us put this $A_j$ in the first position by a
cyclic permutation. We will represent the right hand side vector $\mathbf{v}_n$ in the
form

$$\mathbf{v}_n=\mathbf{v}_{n-1}+ \kappa \mathbf{v}_2\, ,$$
where $\mathbf{v}_{n-1}$ corresponds to the cycle of the length $n-1$, $A_2 \to \ldots
A_n \to A_2$, with the rate constants $q_{i+1 i}=1 /p_i^*$ (and the cyclic convention
$n+1=2$), $\mathbf{v}_2$ corresponds to the cycle of the length 2, $A_1\rightleftharpoons
A_2$, with the rate constants $q_{21} = 1/{p^*_1}$, $q_{12} = 1/{p^*_ 2}$, and $\kappa
\geq 0$. Both velocities $\mathbf{v}_{n-1}$ and $\mathbf{v}_2$ should be calculated for
the same distribution $P$.

We find the constant $\kappa$ from the conditions: $\mathbf{v}_n=\mathbf{v}_{n-1}+\kappa
\mathbf{v}_2$ at the point $P$, hence, the two following reaction schemes, (a) and (b),
should have the same velocities, $\D P/\D t$:
$$\mbox{(a)  }A_n { \overset{1/p_n^*}{\rightarrow}}A_1{ \overset{1/p_1^*}{\rightarrow}}A_2 \mbox{  and  (b)  }
A_n { \overset{1/p_n^*}{\rightarrow}}A_2 \, ; A_1
\underset{\kappa/p_2^*}{\overset{\kappa/p_1^*}{\rightleftharpoons}}A_2 \, .$$
From this condition,
$$\kappa=\left(\frac{p_{n}}{p_{n}^*}-\frac{p_{1}}{p_{1}^*}\right)
\left(\frac{p_{2}}{p_{2}^*}-\frac{p_{1}}{p_{1}^*}\right)^{-1}\, .$$ The
inequality $\kappa \geq 0$ holds because ${p_{1}}/{p^*_{1}}$ is the minimal value of the
flux ${p_{j}}/{p^*_{j}}$.

We just delete the vertex with the smallest outgoing flux from the initial cycle of
length $n$ and add a cycle of the length 2 with the same equilibrium. Let us repeat this
operation for the remaining cycle of the length $n-1$, and so on. At the end, the left
hand side vector $\mathbf{v}_n$ will be represented as the combination with positive
coefficients the vectors $\D P /\D t$ for the cycles of the length 2, $A_i
\rightleftharpoons A_j$ with the same equilibrium. This is the system with detailed
balance. We have to stress here that the set of these transitions and the coefficients
$\kappa$ depend on the current distribution $P$.

For every distribution $P$, the velocity $\D P/\D t$ of every cycle with equilibrium
$P^*$ is a combination with positive coefficients of the velocities for some cycles of
the length two $A_i \rightleftharpoons A_j$ with the same equilibrium. Therefore, the
right hand side of the Kolmogorov equation for any Markov chain with equilibrium $P^*$
also allows such a decomposition.

It is necessary to stress that the decomposition of the right hand side of the Kolmogorov
equation (\ref{MAsterEq0}) into a conic combination of cycles of length 2 depends on the
ordering of the ratios $p_i/p_i^*$ and cannot be performed for all values of $P$
simultaneously.
\end{proof}

For more details and further references see \cite{GorbanEq2012arX}.

\end{document}